\documentclass
[
a4paper,
pdftex,
onecolumn,
DIV=11,
abstracton
]
{scrartcl}

\usepackage
{
	graphicx,
	amssymb,
	amsmath,
	amsthm,
	xcolor,
	dsfont,
	algpseudocode,
	authblk,
	enumerate,
	bm,
	mathtools
}

\usepackage[left=2cm,right=2cm,top=2cm]{geometry}

\usepackage[ruled,vlined]{algorithm2e}
\usepackage[bf]{caption}
\usepackage[colorlinks,
pdffitwindow=false,
plainpages=false,
pdfpagelabels=true,
pdfpagemode=UseOutlines,
pdfpagelayout=SinglePage,
bookmarks=false,
colorlinks=true,
hyperfootnotes=false,
linkcolor=blue,
citecolor=green!50!black]
{hyperref}

\DeclareMathAlphabet{\mathpzc}{OT1}{pzc}{m}{it}

\newcommand{\subfiguretitle}[1]{{\scriptsize{#1}} \\[1mm] }
\newcommand{\R}{\mathbb{R}}

\newcommand{\bQ}{\bm{Q}}
\newcommand{\bL}{\bm{L}}
\newcommand{\pd}[2]{\frac{\partial#1}{\partial#2}}

\providecommand{\norm}[1]{\left\lVert #1 \right\rVert}
\providecommand{\grad}{\nabla}

\DeclareMathOperator{\tr}{tr}
\DeclareMathOperator{\diag}{diag}

\DeclareMathOperator{\mvec}{vec}

\newtheorem{theorem}{Theorem}[section]

\newtheorem{lemma}[theorem]{Lemma}

\newtheorem{definition}[theorem]{Definition}
\newtheorem{remark}[theorem]{Remark}
\newtheorem{example}[theorem]{Example}

\makeatletter
\renewcommand*\env@matrix[1][*\c@MaxMatrixCols c]{%
	\hskip -\arraycolsep
	\let\@ifnextchar\new@ifnextchar
	\array{#1}}
\makeatother

\begin{document}


\title{Continuous Relaxations for the Traveling Salesman Problem}

\author[1]{Tuhin Sahai}
\author[2]{Adrian Ziessler}
\author[3]{Stefan Klus}
\author[2]{Michael Dellnitz}
\affil[1]{\normalsize United Technologies Research Center, 2855 Telegraph Ave, Suite 410,\newline Berkeley, CA, 94705, USA.}
\affil[2]{\normalsize Department of Mathematics, Paderborn University,\newline Warburger Stra\ss e 100, 33098 Paderborn, Germany.}
\affil[3]{\normalsize Department of Mathematics and Computer Science, Freie Universit\"{a}t Berlin,\newline Arnimallee 9, 14195 Berlin, Germany.}

\maketitle

\begin{abstract}
In this work, we aim to explore connections between dynamical systems techniques and combinatorial optimization problems. In particular, we construct heuristic approaches for the traveling salesman problem (TSP) based on embedding the relaxed discrete optimization problem into appropriate manifolds. We explore multiple embedding techniques -- namely, the construction of new dynamical systems on the manifold of orthogonal matrices and associated Procrustes approximations of the TSP cost function. Using these dynamical systems, we analyze the local neighborhood around the optimal TSP solutions (which are equilibria) using computations to approximate the associated \emph{stable manifolds}. We find that these flows frequently converge to undesirable equilibria. However, the solutions of the dynamical systems and the associated Procrustes approximation provide an interesting biasing approach for the popular Lin--Kernighan heuristic which yields fast convergence.  The Lin--Kernighan heuristic is typically based on the computation of edges that have a ``high probability'' of being in the shortest tour, thereby effectively pruning the search space. Our new approach, instead, relies on a natural relaxation of the combinatorial optimization problem to the manifold of orthogonal matrices and the subsequent use of this solution to bias the Lin--Kernighan heuristic. Although the initial cost of computing these edges using the Procrustes solution is higher than existing methods, we find that the Procrustes solution, when coupled with a homotopy computation, contains valuable information regarding the optimal edges. We explore the Procrustes based approach on several TSP instances and find that our approach often requires fewer $k$-opt moves than existing approaches. Broadly, we hope that this work initiates more work in the intersection of dynamical systems theory and combinatorial optimization.

\end{abstract}


\section{Introduction}
\label{sec:Introduction}
\noindent The use of dynamical systems based methods for analyzing optimization algorithms is a burgeoning area of interest. For example, it has been found that if one embeds sufficiently hard instances of the satisfiability problem~\cite{cit:sat_book} into a corresponding dynamical system, one observes transient chaos~\cite{cit:opt_chaos}. In the continuous optimization setting, accelerated momentum methods were analyzed using dynamical systems and calculus of variation based approaches, providing intuitive insight into convergence properties~\cite{cit:diff,cit:wibisono}. However, concrete examples of the direct application of dynamical systems and continuous processes to construct new state-of-the-art algorithms are limited.
In this work, we aim to use dynamical systems and their associated manifolds of the traveling salesman problem (TSP) to extract computational and algorithmic insights.

The TSP is an iconic NP-hard problem that has
received decades of interest~\cite{Cit:cook}. This combinatorial optimization problem arises in a wide
variety of applications related to genome map
construction~\cite{Cit:geneTSP}, telescope management~\cite{Cit:telescope,
Cit:telescope2}, and drilling circuit boards~\cite{Cit:circuit-boards}.
The TSP also naturally arises in applications related to target
tracking~\cite{Cit:target_tracking}, vehicle
routing~\cite{Cit:vehicle_routing}, and communication networks~\cite{Cit:sonet_rings} to name a few.
Recently, a history dependent TSP was
used to construct efficient techniques for learning the structure of Bayesian
networks~\cite{Cit:Tuhin_BN}. For further information about applications related to the TSP, we refer the reader to~\cite{Cit:cook}.

In its basic form, the statement of the TSP is
exceedingly simple. The task is to find the shortest Hamiltonian circuit
through a list of cities, given their pairwise distances. Despite its
simplistic appearance, the underlying problem is NP-hard~\cite{TSP-NP-hard}.
Several heuristics have been developed over the years to solve the problem
including ant colony optimization~\cite{Cit:Ant}, cutting plane
methods~\cite{Cit:cuttingplane2,Cit:cuttingplane}, Christofides heuristic
algorithm~\cite{Cit:christofides}, and the Lin--Kernighan
heuristic~\cite{Cit:LinKernighan}.

In this work, we concentrate on exploring novel orthogonal relaxation and embedding based approximations to the TSP that are inspired from dynamical systems theory. In the first part, we construct a dynamical systems approach for computing locally optimal solutions of the TSP. This flow on the manifold of orthogonal matrices converges to a permutation matrix that minimizes the tour length. Although the method is interesting and elegant, the flow often converges to local minima. For TSP instances with more than $50$ cities, these minima are not competitive when compared to state-of-the-art heuristics~\cite{Cit:cuttingplane2,Cit:cook,Cit:LinKernighan,Cit:cuttingplane}.

However, inspired by this continuous relaxation, we compute the solution to a two-sided orthogonal Procrustes problem~\cite{Cit:Procrustes-book} that relaxes the TSP to the manifold of orthogonal matrices. We find that this Procrustes approach can be combined with the Lin--Kernighan heuristic~\cite{Cit:LinKernighan}
for computing solutions of the TSP. The Lin--Kernighan heuristic is an extremely popular method for the TSP and has been credited with finding the best known solutions for several large instances~\cite{Hel98,Cit:keld2}. It has been particularly successful in finding the best known solutions for several asymmetric TSPs~\cite{Hel98}. We provide a detailed description of the Lin--Kernighan heuristic in Section~\ref{sec:Lin--Kernighan}.

Helsgaun's software package LKH~\cite{Hel98} is a highly successful software implementation of the approach.
This implementation uses minimum spanning trees~\cite{HK70, HK71} to pre-compute candidate sets that contain edges that are likely to be a part of the optimal solution. This biasing methodology is found to reduce the number of $ k $-opt moves compared to baseline minimum tree based methods~\cite{Hel98}. In our work,
the Procrustes solution is used to bias the Lin--Kernighan heuristic algorithm to pick edges are that more likely to be in the optimal tour. We remark that our approach is tightly connected to spectral methods for graphs~\cite{Cit:klus2014spectral}.
Although, the Procrustes based methodology has a higher overall computational cost $O(n^3)$ due to the required eigenvector computations  -- compared to $O(n^{2.2})$ in the case of the traditional Lin--Kernighan heuristic~\cite{Cit:LinKernighan}. However, the Procrustes based Lin--Kernighan computation frequently converges faster (in fewer iterations) than the $1$-tree based approach.

Our goals are twofold: First, we would like to demonstrate that the spectral structure of the associated matrices of these classes of problems contain valuable information that can be exploited for analysis and construction of novel heuristics. Second, we envision that by approximating the spectral structure~\cite{cit:schafer2017owhadi}, one could potentially construct competitive methods for the TSP and the quadratic assignment problem (QAP). Moreover, we note that although eigenvalue and orthogonal approximations have been constructed for the TSP, they have traditionally been used for deriving bounds for the solutions~\cite{Cit:qap_trace,cit:qap_bounds}. To the best of the authors knowledge, this work is the first attempt to use the orthogonal relaxations and \emph{dynamical systems theory} to construct computational methods for the TSP.

Our paper is organized as follows: We start with the  mathematical formulation of the TSP in section~\ref{sec:TSP}. In section~\ref{sec:Lin--Kernighan}-\ref{ssec:subgradientOptimization}, we describe the standard Lin--Kernighan heuristic along with techniques to limit the search space using $ \alpha $-nearness values (based on minimum spanning trees). In section~\ref{sec:Dyn-sys}, we construct dynamical systems on the manifold of orthogonal matrices that converge to Hamiltonian cycles. Using these dynamical systems, we analyze the stability and subsets of the stable manifold of the optimal TSP solutions using set-oriented numerical methods implemented in the software package GAIO \cite{DFJ01}. We perform the computations in an effort to gain insights into the local dynamics of the flow in the neighborhood of the optimal solutions. In general, we find that the basin of attraction is typically quite small and therefore, these dynamical systems converge to undesirable local minima. Inspired by these insights, we use a Procrustes-based approach for biasing the Lin--Kernighan heuristic based on ``$P$-nearness values'' in section~\ref{sec:Procrustes}. Numerical results are presented in section~\ref{sec:Results}. Finally, we conclude with future work in section~\ref{sec:Conclusion}.

\section{The traveling salesman problem}
\label{sec:TSP}

Given a list of $ n $ cities $ \{ C_{1}, C_{2}, \dots, C_{n} \} $ and the associated distances between cities $ C_i $ and $ C_j $, denoted by $ d_{ij} $, the TSP aims to find an ordering $ \sigma $ of $ \{ 1, 2, \dots, n \} $ such that the tour cost, given by
\begin{equation} \label{eq:basic_cost}
    c = \sum_{i=1}^{n-1} d_{\sigma(i), \sigma(i+1)} + d_{\sigma(n), \sigma(1)},
\end{equation}
is minimized. For the Euclidean TSP, for instance, $ d_{ij} = \norm{x_i - x_j}_2 $, where $ x_i \in \R^d $ is the position of $ C_i $. In general, however, the distance matrix $ D = (d_{ij}) $ does not have to be symmetric. The ordering $ \sigma $ can be represented as a unique permutation matrix $ P $. Note, however, that due to the underlying cyclic symmetry, multiple orderings -- corresponding to different permutation matrices -- have the same cost.

There are several equivalent ways to define the cost function of the TSP. We restrict ourselves to the trace\footnote{The trace of a matrix $ A \in \R^{n \times n} $ is defined to be the sum of all diagonal entries, i.e., $ \tr(A) = \sum_{i=1}^n a_{ii} $.} formulation proposed in~\cite{Won95}. Let $ \mathcal{P}_n $ denote the set of all $ n \times n $ permutation matrices, then the TSP can be written as a combinatorial optimization problem of the form
\begin{equation} \label{eq:TSPCost}
    \min_{P \in \mathcal{P}_n} \tr \left( A^T P^T B P \right),
\end{equation}
where $ A = D $ and $ B = T $. Here, $ T $ is defined to be the adjacency matrix of the cycle graph of length $ n $. 
In what follows, we use the undirected cycle graph adjacency matrix for symmetric TSPs and the one corresponding to the directed cycle graphs for asymmetric TSPs. The matrices are defined as,
\begin{equation*} \label{eq:tmatrix}
    T_\text{dir} = \begin{pmatrix}
        0 & 1 &   &  &   \\
          & 0 & 1 &  &   \\
          &   & \ddots & \ddots &  \\
          &   &   & 0  & 1 \\
        1 &   &   &  & 0
    \end{pmatrix} \; \text{ or } \;
    T_\text{undir} = \begin{pmatrix}
        0 & 1 &   &  & 1 \\
        1 & 0 & 1 &  &   \\
          & \ddots  & \ddots & \ddots &  \\
          &   & 1 & 0  & 1 \\
        1 &   &   & 1 & 0
    \end{pmatrix}.
\end{equation*}
The equivalence of \eqref{eq:basic_cost} and \eqref{eq:TSPCost} can be derived easily using the observation that $ \tilde{T} \coloneqq P^T T P$ is a permuted tour matrix, i.e., the $(i,j)$ entry is $1$ if the tour goes from city $C_i$ to city $C_j$. Thus, for any permutation, $\tr(D^T \tilde{T}) = \sum_{i,j=1}^n d_{ij} \tilde{t}_{ij} $ is simply the sum of the distances associated with each edge. In what follows, we restrict our work to symmetric matrices; thus, we simply consider tour matrix $ T = T_\text{undir} $ for undirected graphs.

The TSP can also be regarded as a special case of the general QAP \cite{BCPP98,KB57}, given by
\begin{equation} \label{eq:qapcost}
    \min_{P \in \mathcal{P}_n} \tr\left(A^{T}P^{T} BP + P^{T}C\right),
\end{equation}
or, alternatively as, a special case of the graph matching problem. The relationship between various combinatorial optimization problems is explored in Figure~\ref{fig:problems}. In order to convert the minimization problem into an equivalent maximization problem, note that
\begin{equation} \label{eq:MinMaxEquivalence}
    \begin{split}
        \norm{A - P^T B P}_F^2 &= \tr\left( A^T A \right)
                                - 2 \tr\left( A^T P^T B P \right)
                                + \tr\left(B^T B\right) \\
                               &= \norm{A}_F^2 - 2 \tr\left(A^T P^T B P \right) + \norm{B}_F^2.
    \end{split}
\end{equation}
Thus, the norm is minimized if the trace is maximized and vice versa. We note that similar trace formulations for the QAP were derived in~\cite{Cit:qap_trace}.

\begin{figure}
    \includegraphics[width=\textwidth]{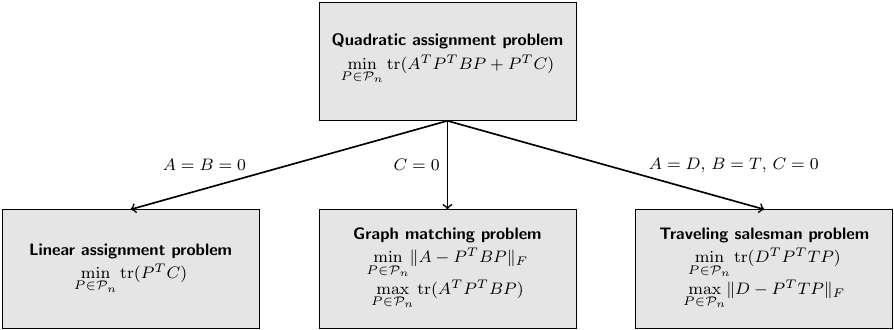}
    \caption{Relationships between the various combinatorial optimization problems.}
    \label{fig:problems}
\end{figure}

Over the last few decades, a plethora of heuristics has been developed to solve the TSP efficiently. In order to find a good approximation of the optimal tour, typically different global and local heuristics are combined. A very efficient and powerful methodology is to construct an initial solution with the aid of greedy algorithms, for instance the nearest neighbor heuristic, and to improve the solution successively using local heuristics such as $ k $-opt move based methods (as described in section~\ref{sec:Lin--Kernighan}). As noted previously, one of the best available TSP solvers is Helsgaun's LKH software~\cite{Hel98}. We now provide more details on LKH and its implementation.

\subsection{The Lin--Kernighan heuristic}
\label{sec:Lin--Kernighan}

The Lin--Kernighan heuristic is a popular heuristic for the TSP introduced in~\cite{Cit:LinKernighan}. Starting from an initial tour, the approach progresses by extracting edges from the tour and replacing them with new edges, while maintaining the Hamiltonian cycle constraint. If $k$ edges in the tour are simultaneously replaced, this is known as the $k$-opt move~\cite{Cit:helsgaun_k_opt}. To prune the search space, the algorithm relies on minimum spanning trees~\cite{HK70,HK71} to identify edges that are more likely to be in the tour. This ``importance'' metric for edges is called $\alpha$-nearness and described subsequently. The algorithm has found great success on large instances of the TSP~\cite{Hel98,Cit:keld2}.
Note that this algorithm has been extended to generalized TSPs~\cite{Cit:generalized_tsp} and clustered TSPs~\cite{Cit:tsp_clustered}.

The LKH package \cite{Hel98,Cit:keld2} offers different heuristics to compute an initial tour. The standard method is to choose one node at random and to iteratively add edges based on computed $\alpha$-nearness values and related candidate sets until a tour is found.
When an initial tour has been found, LKH improves it using local heuristics. A very popular and efficient local heuristic is the $ k $-opt move. The simplest version, $ 2 $-opt, removes two edges of the tour and reconnects the subtours as shown in Figure~\ref{fig:2-opt and 3-opt}. If the resulting tour is shorter than the original tour, the step is accepted and rejected otherwise. Similarly, $ 3 $-opt removes three edges of the current tour, reconnects the subtours and picks the shortest tour.

\begin{figure}[htb]
	\centering
	\begin{minipage}[b]{0.49\textwidth}
		\centering
		\subfiguretitle{a)}
		\includegraphics[scale=0.29]{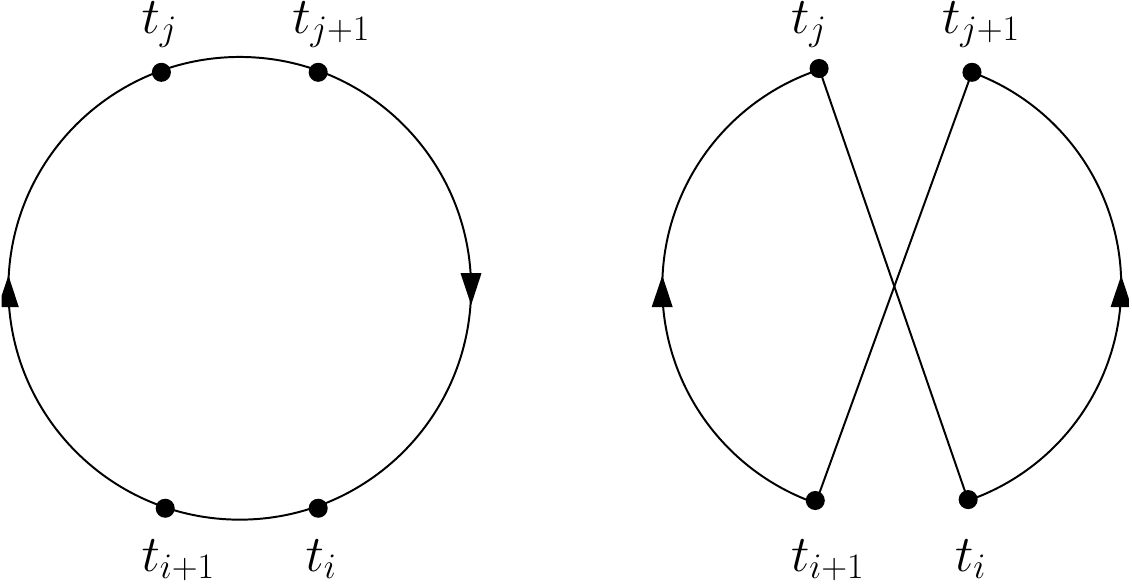}
	\end{minipage}
	\begin{minipage}[b]{0.49\textwidth}
		\centering
		\subfiguretitle{b)}
		\includegraphics[scale=0.29]{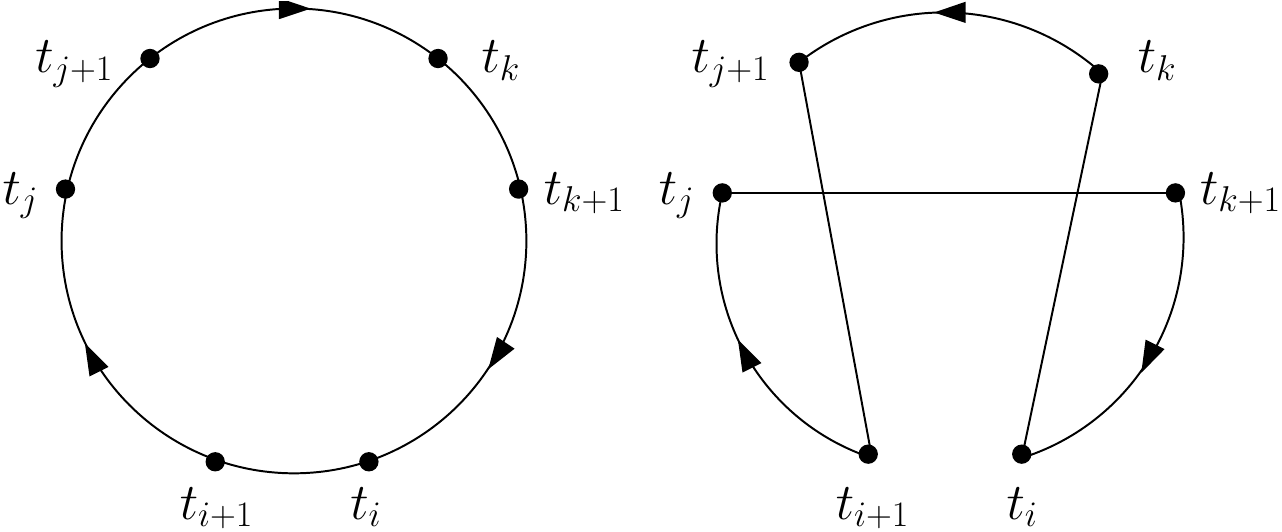}
	\end{minipage}
	\caption{2-opt and 3-opt move.}
	\label{fig:2-opt and 3-opt}
\end{figure}

\subsection{Candidate sets and $\alpha$--nearness}\label{sec:candidate-alpha}
LKH uses $ k $-opt with varying $ k $. The basic move is a sequential $ 5 $-opt step. In order to limit the search space and to increase the efficiency of $ k $-opt moves, candidate sets that contain promising edges are computed for all cities. Methods to construct candidate sets for large TSPs have to be efficient in terms of both CPU time and memory usage. As mentioned previously, the standard approach implemented in LKH, called $ \alpha $-nearness, is based on minimum spanning trees or, to be more precise, on 1-trees (a slight variant of minimum spanning trees).

\begin{definition}
	Let $ G = (V, E) $ be a graph with vertices $ V $ and edges $ E $. A 1-tree for $ G $ is defined to be a spanning tree for the vertices $ V \setminus \{v_1\} $ plus two additional edges $ e \in E $ incident to vertex $ v_1 $.
\end{definition}

A 1-tree with minimum weight is called a minimum 1-tree. Note that every tour is a 1-tree with the additional property that the degree of each vertex is two. It has been found that the minimum 1-tree typically contains several edges that lie in the optimal tour. The definition of $ \alpha $-nearness is based on a sensitivity analysis using 1-trees. The $ \alpha $-nearness value for the cities $ C_i $ and $ C_j $ is, roughly speaking, the difference between a minimum 1-tree and a 1-tree that is required to contain edge $ (v_i, v_j) $. That is, if an edge belongs to the minimum 1-tree, then the $ \alpha $-nearness value is $ 0 $ and the edge is assigned a high probability of being part of the shortest tour.

For each city, the candidate set is then defined to be the set of the $ m $ incident edges with the lowest $ \alpha $-nearness values. The candidate sets are used to limit and direct the search. Candidate sets based solely on the distance between cities are typically not connected (for instance, see Figure~\ref{fig:distance_vs_procrustes}) and convergence to a good solution of the TSP is expected to be slow. It was shown by Stewart~\cite{Ste87} that minimum spanning trees, which are by definition always connected, can be used to increase the efficiency of local heuristics.

\subsection{Subgradient optimization}\label{ssec:subgradientOptimization}
The $\alpha$-nearness values, typically do not give rise to optimal tours when coupled with k-opt moves. In order to improve the $ \alpha $-nearness values, a subgradient optimization method is often used. This method modifies the original distance matrix $ D $ in a way such that the degree of almost all vertices of the optimized 1-tree converge to a value of $ 2 $. The entries of the new distance matrix, $ \widetilde{D}(\pi) $, are computed as
\begin{equation*}
\widetilde{d}_{ij}(\pi) = d_{ij} + \pi_i + \pi_j.
\end{equation*}
The $ \pi $ values, sometimes called penalties, change the distances between the cities. The basic idea is to make edges incident to vertices with a low degree shorter and edges incident to vertices with a high degree longer so that the resulting 1-tree is close to a tour. This transformation of the distance matrix does not change the shortest tour and leads to significantly improved $ \alpha $-nearness values~\cite{Hel98}. This method can also be used to compute a lower bound which is in general very close to the optimal tour length~\cite{HK70, HK71}. Figure~\ref{fig:SubgradientOpt} shows the impact of the subgradient optimization. In the example, most of the edges of the optimal tour are already present in the optimized 1-tree. For a more detailed description of $ \alpha $-nearness, 1-trees, and the subgradient optimization scheme, we refer to~\cite{HK70, HK71, Hel98}.

\begin{figure}[htb]
	\centering
	\begin{minipage}[t]{0.32\textwidth}
		\centering
		\subfiguretitle{a)}
		\includegraphics[width=\textwidth]{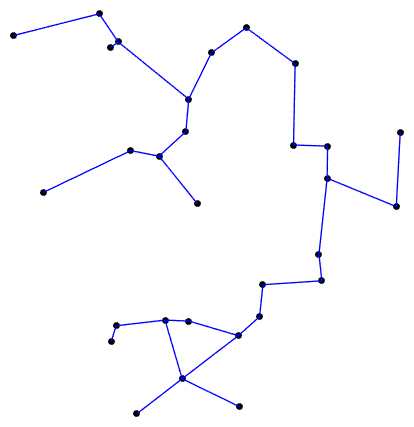}
	\end{minipage}
	\begin{minipage}[t]{0.32\textwidth}
		\centering
		\subfiguretitle{b)}
		\includegraphics[width=\textwidth]{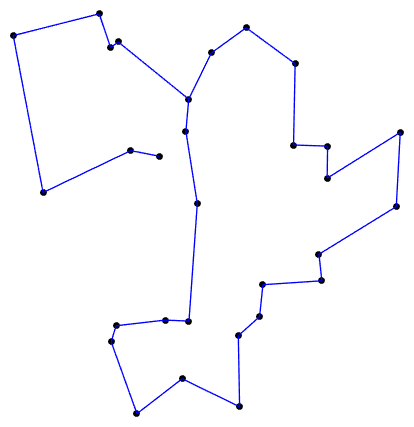}
	\end{minipage}
	\begin{minipage}[t]{0.32\textwidth}
		\centering
		\subfiguretitle{c)}
		\includegraphics[width=\textwidth]{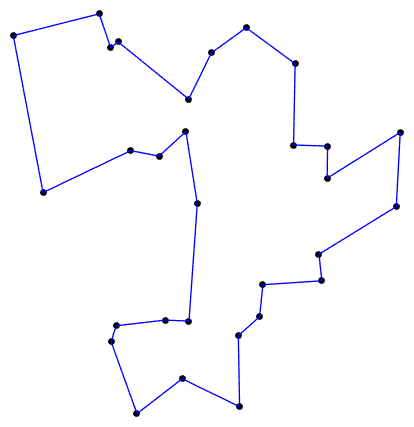}
	\end{minipage}
	\caption{Impact of the subgradient optimization. a) Minimal 1-tree of the original distance matrix.
		b) Minimal 1-tree of the transformed distance matrix. c) Shortest tour.}
	\label{fig:SubgradientOpt}
\end{figure}

\section{Dynamical systems approach}
\label{sec:Dyn-sys}

In this section, we construct a dynamical systems approach for computing optimal tours for the TSP. In particular, we use matrix differential equations defined on the manifold of orthogonal matrices. As mentioned in previous sections, solutions of the TSP can be represented as permutation matrices. It is well known that permutation matrices lie on the manifold of orthogonal matrices. Our goal is to construct flows that minimize the TSP cost as they evolve. Note that gradient flow methods were first used by Brockett to compute eigenvalues and to solve linear programming or least squares matching problems~\cite{Bro89, Bro91}. This approach was subsequently also used for combinatorial optimization problems~\cite{WY10,Won95, ZP08}. We will now formulate multiple cost functions for constructing gradient flows for the TSP. In what follows, let us denote by ${\mathcal{O}_n = \left\{ P \in \R^{n \times n} \mid P^T P = I \right\} \supset \mathcal{P}_n }$ the set of all $ n \times n $ orthogonal matrices. In this section, we consider the  following orthogonal relaxation of the combinatorial optimization problem~\eqref{eq:TSPCost},
\begin{equation} \label{eq:RelTSPCost}
\min_{P \in \mathcal{O}_n} \tr \left( A^T P^T B P \right).
\end{equation} Given that we will use the solution of the \emph{Procrustes problem} in both the dynamical systems approach and Lin-Kernighan heuristic, we start by defining the problem and its solution.

\subsection{The two-sided orthogonal Procrustes problem}
\label{sec:twosidedproc}
We start by considering the standard formulation of the Procrustes problem. We will then modify it to the TSP setting. Let $ A $ and $ B $ be two symmetric $ n \times n $ matrices. Then
\begin{equation}
    \min_{P \in \mathcal{O}_n} \norm{ A - P^T B P }_F
    \label{eq:TSOPP}
\end{equation}
is called the \emph{two-sided orthogonal Procrustes problem}. As shown in \eqref{eq:MinMaxEquivalence}, cost function \eqref{eq:RelTSPCost} is minimized if the cost function $\norm{ A - P^T B P }_F$ of the Procrustes problem is maximized and vice versa. Since $ \mathcal{P}_n \subset \mathcal{O}_n $, the cost of the orthogonal matrix is always lower than (or equal to if the matrices $ A $ and $ B $ are permutation-similar) the cost of the permutation matrix.

\begin{theorem}
Given two symmetric matrices $ A $ and $ B $, whose eigenvalues are distinct, let $ {A = V_A \Lambda_A V_A^T} $ and $ B = V_B \Lambda_B V_B^T $ be eigendecompositions, with $ \Lambda_A = \diag\left(\lambda_{A}^{(1)}, \dots, \lambda_{A}^{(n)}\right) $, $ \Lambda_B = \diag\left(\lambda_{B}^{(1)}, \dots, \lambda_{B}^{(n)}\right) $, and $ \lambda_{A}^{(1)} \geq \dots \geq \lambda_{A}^{(n)} $ as well as $ \lambda_{B}^{(1)} \geq \dots \geq \lambda_{B}^{(n)} $. Then every orthogonal matrix $ P^* $ which minimizes~\eqref{eq:TSOPP} has the form
\begin{equation*}
    P^* = V_B S V_A^T,
\end{equation*}
where $ S = \diag(\pm 1, \dots, \pm 1) $.
\end{theorem}

A proof of this theorem can be found in~\cite{Sch68}, for example. If the eigenvalues of $ A $ and $ B $ are distinct, then there exist $ 2^n $ different solutions with the same cost. If one or both of the matrices possess repeated eigenvalues, then the eigenvectors in the matrices $V_A$ and $V_B$ are determined only up to basis rotations, which further increases the solution space.

We note that, as shown in~\eqref{eq:MinMaxEquivalence}, the minimization of the TSP cost corresponds to the maximization of the cost in~\eqref{eq:TSOPP} and not the standard minimization found in literature. The theorem states that in order to minimize the cost function in~\eqref{eq:TSOPP}, the eigenvalues and corresponding eigenvectors have to be sorted both in either increasing or decreasing order. On the other hand, from the proof of the theorem it can be seen that in order to compute the solution of~\eqref{eq:RelTSPCost}, the eigenvalues and eigenvectors of $ A $ and $ B $, respectively, have to be sorted in opposite order (which corresponds to the maximization of the cost in~\eqref{eq:TSOPP}). A similar condition was noted in~\cite{AW00}.

Let us now consider gradient flows for orthogonal and tour matrices. We note that the Procrustes problem is relevant for the resulting dynamical systems as demonstrated below.

\subsection{Gradient flows for orthogonal matrices}\label{ssec:orthogonal_flow}

The orthogonal relaxation of the combinatorial optimization problem \eqref{eq:TSPCost}, given by \eqref{eq:RelTSPCost}, can be solved using a steepest descent method on the manifold of orthogonal matrices. Given a cost function $ F $, the gradient flow is defined as
\begin{equation} \label{eq:GradFlow}
    \dot{P} = -\grad{F(P)},
\end{equation}
which is a matrix differential equation evolving on the manifold of orthogonal matrices. That is, starting with an orthogonal matrix $ P $, the trajectory remains for all time in~$\mathcal{O}_n $. Let $ [A, B] = A B - B A $ be the standard Lie bracket and $ \{A, B\} = A^T B - B^T A $ the generalized Lie bracket. The gradient of a function $ F $ defined on the manifold of orthogonal matrices is
\begin{equation} \label{eq:grad F}  
    \grad{F}(P) = F_P - P F_P^T P = P \{ P, F_P \},
\end{equation}
where $ F_P $ is the matrix of partial derivatives~\cite{EAS98}, i.e., $ (F_P)_{ij} = \pd{F}{P_{ij}} $.

\begin{lemma} \label{lem:relaxed QAP}
For the cost function $ F(P) = \tr\left( A^T P^T B P \right) $, we obtain
\begin{equation*}
    \grad{F}(P) = P \left( \left\{ P^T B P, A \right\} + \left \{P^T B^T P, A^T \right \}\right).
\end{equation*}
\end{lemma}
\begin{proof}
Since $ F_P = B P A^T + B^T P A $, see~\cite{PP08}, using \eqref{eq:grad F} it directly follows that $\grad{F}(P) = P \{ P, B P A^T + B^T P A \}$, which can be rewritten as above.
\end{proof}

This is our generalization of the matrix flow defined in~\cite{ZP08} for the symmetric graph matching problem. If $ A $ and $ B $ are symmetric, this can be simplified to
\begin{equation} \label{eq:SQAPFlow}
    \grad{F}(P) = 2 P \left[ P^T B P, A \right].
\end{equation}
Since the optimal solution of this optimization problem is in general not a permutation matrix, Zavlanos and Pappas~\cite{ZP08} use a second term for solving the graph matching problem, which penalizes nonnegative entries. Note that the set of permutation matrices is the intersection of the sets of orthogonal and nonnegative matrices. In order to force the gradient flow to converge to a permutation matrix, a cubic penalty function is used.

\begin{lemma} \label{lem:penalty function}
Let $ \circ $ denote the Hadamard or element-wise product of two matrices. For the penalty function $ G(P) = \tfrac{1}{3} \tr \left( P^T \left( P - (P \circ P) \right) \right) = \tfrac{1}{3} n - \tfrac{1}{3} \sum_{i,j=1}^n p_{ij}^3 $, the gradient is given by
\begin{equation} \label{eq:PenaltyFlow}
    \grad{G}(P) = P \left( (P \circ P)^T P - P^T(P \circ P) \right).
\end{equation}
\end{lemma}
\begin{proof}
Note that $ G_P = -(P \circ P) $. Using \eqref{eq:grad F} results in the above gradient.
\end{proof}

By combining the two functions $ F $ and $ G $, it is possible to compute a permutation matrix which is ``close'' to the optimal orthogonal solution. In~\cite{ZP08}, the steady state solution of the superimposed gradient flows for $ F $ and $ G $, given by
\begin{equation*}
    \begin{split}
        \dot{P} = -(1 & - k) P \left( \left\{ P^T B P, A \right\} + \left \{P^T B^T P, A^T \right \}\right)\\
                     & - k \, P \left( (P \circ P)^T P - P^T(P \circ P) \right),
    \end{split}
\end{equation*}
is computed for $ k = 0 $, then the parameter $ k $ is set to a value sufficiently close to $ 1 $ so that the flow converges to a permutation matrix. Another approach is to apply a homotopy-based scheme, where $ k $ is the continuation parameter which is gradually increased until the solution is close to a permutation matrix.

In what follows, we will consider the TSP as a constrained optimization problem of the form,
\begin{equation}\label{eq:minPwithEqualityConstraints}
	\begin{aligned}
	\min_{P \in \mathcal{O}_n} & \; \tr \left( A^T P^T B P \right), \\
	s.t. \; & \; G(P) = 0.
	\end{aligned}
\end{equation}
This formulation gives rise to the following set of equations,
\begin{equation}\label{eq:TSP flow}
\begin{aligned}
    \dot{P}       &= -P \left( \left\{ P^T B P, A \right\} + \left \{P^T B^T P, A^T \right\} \right)
- \lambda P \left( (P \circ P)^T P - P^T(P \circ P) \right), \\
\dot{\lambda} &= \tfrac{1}{3} \tr \left( P^T \left( P - (P \circ P) \right) \right).
\end{aligned}
\end{equation}
The above set of equations are obtained by using gradient descent on the Lagrangian cost function, as described in~\cite{PB88}, for general constrained optimization problems. Note that $\lambda$ in~\eqref{eq:TSP flow} is the Lagrange multiplier.

%
	
\begin{example} \label{ex:P flow}
	In order to illustrate the gradient flow approach, let us consider a simple TSP with 10 cities. Using~\eqref{eq:TSP flow}, we obtain the results shown in Figure~\ref{fig:P flow}. In this example, the dynamical system converges to the optimal tour.
	
	\begin{figure}[!htb]
		\centering
		\begin{minipage}[c]{0.32\textwidth}
			\centering
			\subfiguretitle{a)}
			\includegraphics[width=0.8\textwidth]{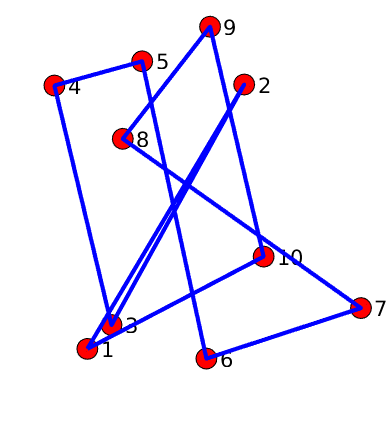}
		\end{minipage}
		\begin{minipage}[c]{0.32\textwidth}
			\centering
			\subfiguretitle{b)}
			\includegraphics[width=0.8\textwidth]{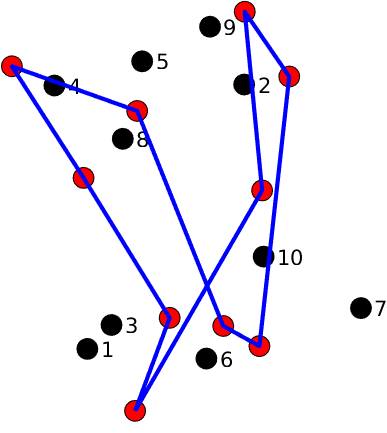}
		\end{minipage}
		\begin{minipage}[c]{0.32\textwidth}
			\centering
			\subfiguretitle{c)}
			\includegraphics[width=0.8\textwidth]{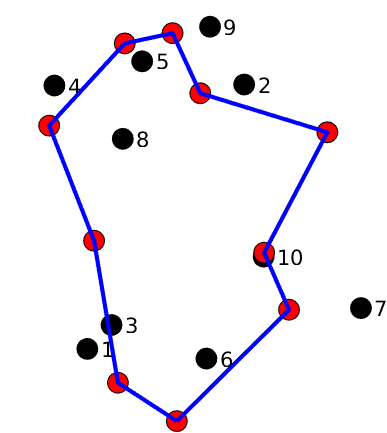}
		\end{minipage} \\[1em]
		\begin{minipage}[c]{0.32\textwidth}
			\centering
			\subfiguretitle{d)}
			\includegraphics[width=0.8\textwidth]{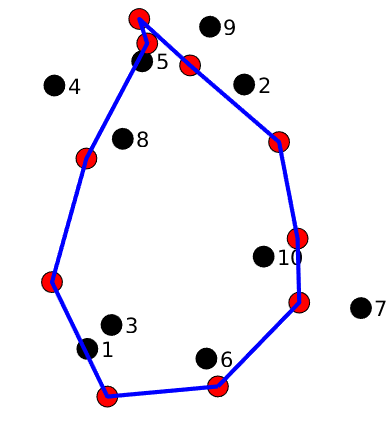}
		\end{minipage}
		\begin{minipage}[c]{0.32\textwidth}
			\centering
			\subfiguretitle{e)}
			\includegraphics[width=0.8\textwidth]{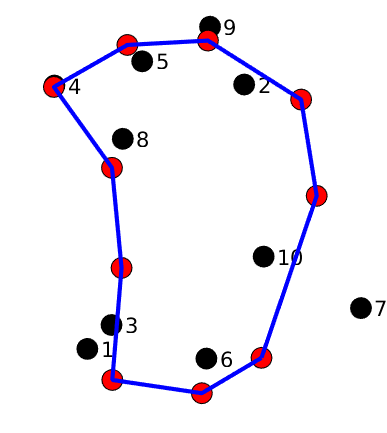}
		\end{minipage}
		\begin{minipage}[c]{0.32\textwidth}
			\centering
			\subfiguretitle{f)}
			\includegraphics[width=0.8\textwidth]{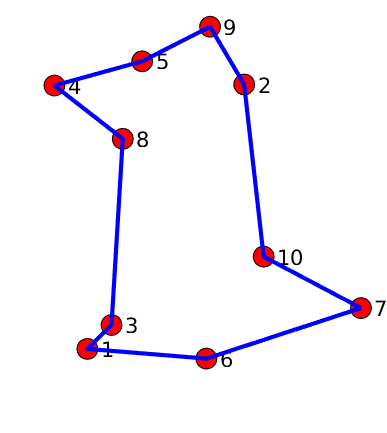}
		\end{minipage}
		\caption{Traveling salesman problem with 10 cities solved using the gradient flow~\eqref{eq:TSP flow}. The original positions of the cities are shown in black, the positions transformed by the orthogonal matrix $ P $ in red. a) Initial trivial tour given by $ \sigma = (1, \dots, 10) $. b--d) Intermediate solutions. e) Convergence to an orthogonal matrix which is ``close'' to a permutation matrix with respect to any matrix norm. f) Extraction of the corresponding permutation matrix. The initial tour was transformed into the optimal tour by the gradient flow.}
		\label{fig:P flow}
	\end{figure}
	
\end{example}
	
Next, we will perform a detailed numerical study of the gradient flow \eqref{eq:TSP flow} for a simple TSP with five cities.
We first consider \eqref{eq:TSP flow} without the equality constraints, i.e.,
\begin{equation}\label{eq:P_flow_noEqualityConstraint}
	\dot{P} = -P \left( \left\{ P^T B P, A \right\} + \left \{P^T B^T P, A^T \right\} \right).
\end{equation}
As mentioned previously, for the symmetric TSP, we set $A = D$ and $B = T_{\text{undir}}$. Solving \eqref{eq:P_flow_noEqualityConstraint} forward in time yields a solution that minimizes the cost function $F(P)$ in $\mathcal{O}_n$. This solution can also be computed by solving the two-sided orthogonal Procrustes problem \eqref{eq:TSOPP}. However, since the matrix $T_{\text{undir}}$ possesses repeated eigenvalues, the Procrustes problem has infinitely many solutions due to rotational degeneracy. An illustration of the \emph{Procrustes sets} is shown in Figure~\ref{fig:procrustesSolutions5Nodes}. We note that the structure of a Procrustes set becomes more and more complex for larger $n \in \mathbb{N}$ since the number of repeated eigenvalues also increases with $n$.

\begin{figure}[!htb]
	\begin{minipage}{0.49\textwidth}
	\includegraphics[width = \textwidth]{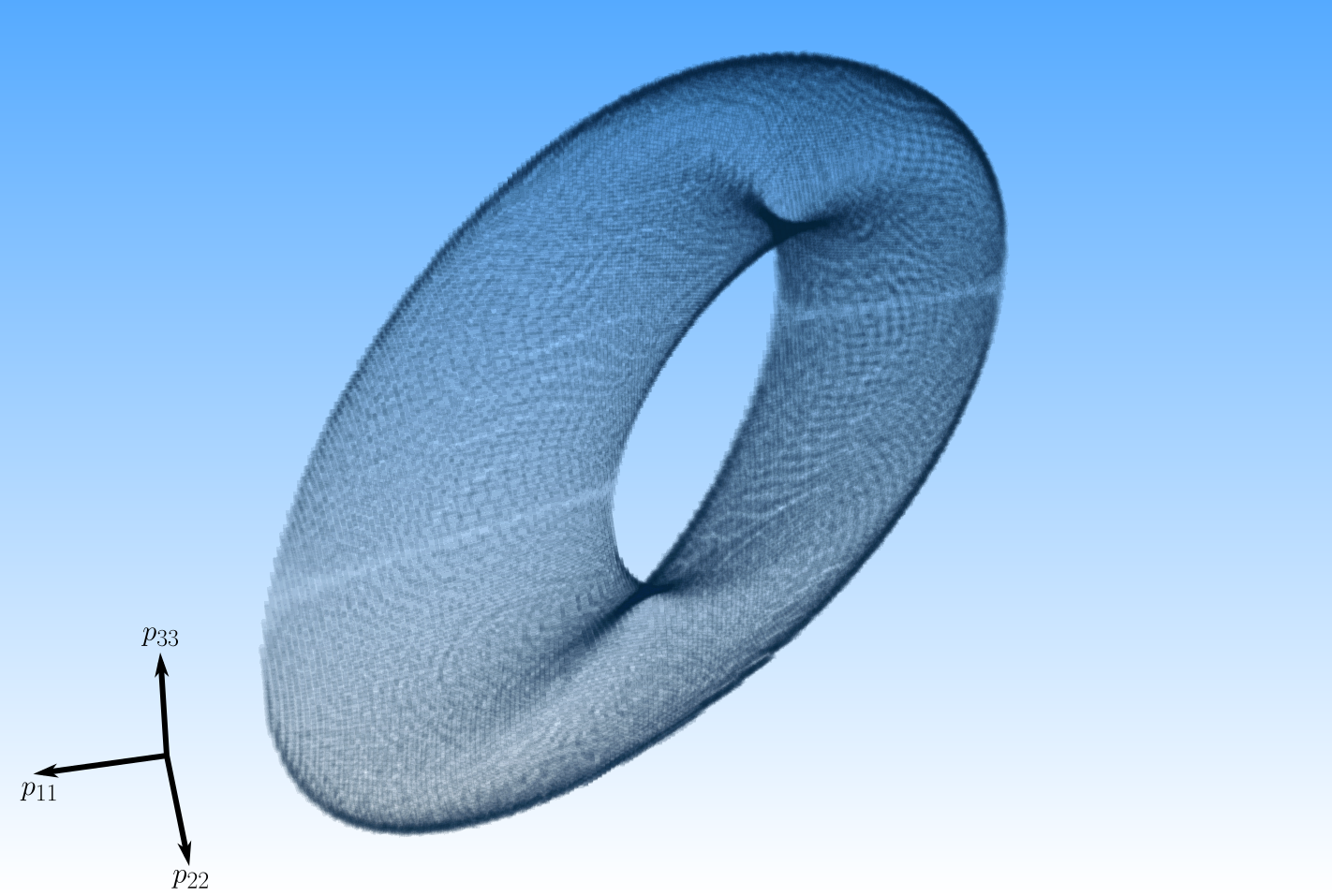}
	\centering \scriptsize{(a)}
\end{minipage}
\hfill
\begin{minipage}{0.49\textwidth}
	\includegraphics[width = \textwidth]{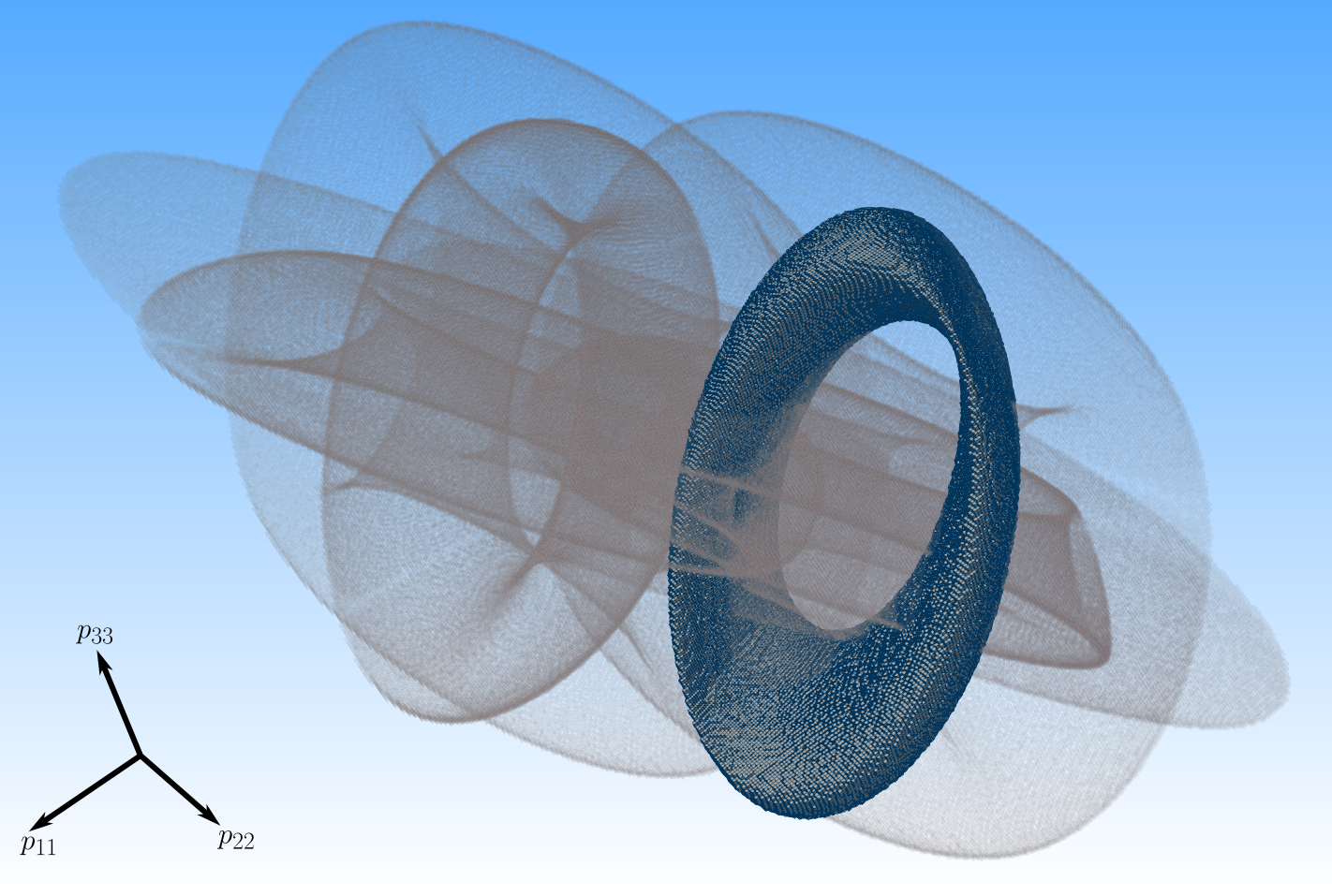}
	\centering \scriptsize{(b)}
\end{minipage}
\caption{(a) Three-dimensional projection of one of the $8$ Procrustes sets. (b) All $8$ procrustes sets in a three-dimensional projection highlighting the Procrustes set from~(a).}
\label{fig:procrustesSolutions5Nodes}	
\end{figure}

To shed light on the stability and local dynamics around the optimal TSP solutions we approximate \emph{subsets of the stable manifold} of the Procrustes solutions such that two permutation matrices are inside these sets. Note that the two matrices are explicitly shown in~\eqref{eq:two_perm_mat}. This numerical study will help us to understand if the Procrustes solutions are robust under small perturbations of the initial permutation matrix and whether or not the Procrustes solution is, in general, a viable approach to constructing useful initial conditions and determine how close the Procrustes solution is to the optimal permutation matrix.
In order to compute the sets of interest, we will use a set-oriented continuation technique developed in \cite{DH96}, which recently has been extended to the computation of unstable manifolds for infinite dimensional dynamical systems \cite{ZDG18}.

Let $\bQ \subset \mathbb{R}^{n^2}$ be a compact set within which we want to approximate the subset of the stable manifold of a permutation matrix $\bar P$. We define a partition $\bL$ of $\bQ$ to be a finite family of compact subsets of $\bQ$ such that
\[
\bigcup_{B \in \bL}B = \bQ\quad \mbox{and}\quad \mbox{int} B\cap \mbox{int} B' = \emptyset,\ \mbox{for all } B, B'\in \bL,~ B\neq B'.
\]
Moreover, we denote by $\bL(x) \in \bL$ the element of $\bL$ containing $x\in \bQ$. In our context, $x$ is a reordering of an (orthogonal) $n \times n$ matrix $\widetilde{P}$ into a vector, i.e., ${x = \mvec(\widetilde{P})}$ (cf.~\cite{PP08}). We consider a nested sequence $\bL_s,\ s \in \mathbb{N}$, of successively finer partitions of $\bQ$, requiring that for all $B \in \bL_s$ there exist cells $B_1,\ldots,B_m \in \bL_{s+1}$ such that ${B = \bigcup_i B_i}$ and ${\mbox{volume}(B_i) = \tfrac{1}{2}\,\mbox{volume}(B)}$.
A set $B \in \bL_s$ is said to be of \emph{level}~$s$.
The partition $\bL_{s+1}$ is computed by a subdivision procedure, where we subdivide each cell of the previous partition $\bL_{s}$ with respect to the $i$-th coordinate, where $i$ is varied cyclically (for more details see \cite{DH97}). However, these subdivision steps are only done virtually and we do not store the cells of the partition $\bL_{s+1}$.

We assume that $C = \bL_s(\mvec(\bar P))$ is the cell which contains the initial permutation matrix $\bar P$ in vector form. Furthermore, let us denote by $\Phi$ the time-$\tau$-map of \eqref{eq:P_flow_noEqualityConstraint}.
Then the numerical realization of the continuation algorithm for the approximation of the subsets of the stable manifold can be described as follows:\\

\begin{algorithm}[H]
	\caption{Approximation of arbitrary subsets of the stable manifold of \eqref{eq:P_flow_noEqualityConstraint}}\label{alg:continuation}
	\vspace{1em}
	Choose an initial cell $\bQ \subset \mathbb{R}^{n^2}$, defined by a $n^2$-dimensional generalized rectangle of the form
	\[
	\bQ(c,r) = \left\{ y \in \mathbb{R}^{n^2}: |y_i - c_i|\leq r_i \mbox{ for } i=1,\ldots,n^2 \right\},
	\]
	where $c,r \in \mathbb{R}^{n^2},\ r_i > 0$ for $i=1,\ldots,n^2$, are the center and the radii, respectively. Choose a partition $\bL_s$ of $\bQ$ and a cell $C \in \bL_s$ such that $\mvec(\bar P) \in C$.
	\begin{enumerate}
		\item Set $ C_0 = \{C\}$.
		\item Continuation step: For $j=0,1,2,\ldots$ define
		\begin{align*}\label{eq:continuation}
		C_{j+1} = \left\lbrace B \in \bL_{s}: \exists B' \in C_j \mbox{ such that }B\cap \left(\mvec\circ \Phi\left(\mvec^{-1}(B')\right)\right) \neq \emptyset\right\rbrace.
		\end{align*}
	\end{enumerate}
\end{algorithm}

\begin{remark}\label{rmk:continuation_algorithm}\quad
	\begin{enumerate}
		\item[a)] In the application of Algorithm~\ref{alg:continuation} we have to perform the continuation step
		\begin{align*}
		C_{j+1} = \left\lbrace B \in \bL_{s}: \exists B' \in C_j \mbox{ such that }B\cap \left(\mvec\circ \Phi\left(\mvec^{-1}(B')\right)\right) \neq \emptyset\right\rbrace.
		\end{align*}
		Numerically this is realized as follows: First, $\Phi$ is evaluated for a large number of test matrices $\mbox{vec}^{-1}(x),\ x \in B'$, for each cell $B' \in C_j$. Then a cell $B \in \mathcal{P}_s$ is added to the collection $C_{j+1}$ if there is at least one $x \in B'$ such that ${\left(\mbox{vec} \circ \Phi\left( \mbox{vec}^{-1}(x)\right)\right) \in B}$. In practice, the components $x\in B'$ of the test matrices can be chosen according to several different strategies: For low-dimensional problems one can choose them from a regular grid within each cell $B$. Alternatively, one can select them via a Monte-Carlo sampling. Observe however that, in general, $\mvec^{-1}(x) \notin \mathcal{O}_n$ for $x \in B'$. Hence, in order to construct an orthogonal matrix we compute a polar decomposition of $\mvec^{-1}(x)$ (cf.~\cite{High86}). The polar decomposition yields a small perturbation of the permutation matrix $\bar P$.
		\item[b)] The number of continuation steps $j$ crucially depends on the choice of the integration time $\tau$. In general, the smaller we choose $\tau$ the more continuation steps are done which leads to a higher computational cost. However, to prevent isolated cells in the covering, $\tau$ has to be chosen sufficiently small.
		\item[c)] The choice of the integration time $\tau$ can be relaxed by choosing a finite time grid $\{t_0,\ldots,t_N\}$, for $t_N = \tau$, where we mark all cells in $\bL_{s}$ which are hit in each time step. This allows us to increase the integration time without decreasing the quality of the covering obtained by Algorithm~\ref{alg:continuation}.
	\end{enumerate}
	
\end{remark}

For the five-dimensional example, we choose $\bQ = [-1,1]^{25}$, $\tau = 10000$ and a fine partition $\bL_s$ of $\bQ$ for $s=175$. In particular, this means that $\bQ$ is subdivided into $2^{175}$ cells of radius $r_i = \frac{1}{2^7}$, for $i=1,\ldots,25$. Moreover, we use the strategy described in Remark~\ref{rmk:continuation_algorithm}~c). In Figure~\ref{fig:mostAttractingSets}, we illustrate two subsets of the stable manifold containing the permutation matrices
\begin{equation}
	P_1 = \begin{pmatrix}	1 & 0 & 0 & 0 & 0 \\ 0 & 0 & 0 & 1 & 0 \\ 0 & 1 & 0 & 0 & 0 \\ 0 & 0 & 1 & 0 & 0 \\ 0 & 0 & 0 & 0 &1 \end{pmatrix} \quad \mbox{and} \quad P_2 = \begin{pmatrix}	0 & 0 & 0 & 0 & 1 \\ 0 & 0 & 1 & 0 & 0 \\ 0 & 1 & 0 & 0 & 0 \\ 0 & 0 & 0 & 1 & 0 \\ 1 & 0 & 0 & 0 & 0 \end{pmatrix}.
\label{eq:two_perm_mat}
\end{equation}
Small perturbations of these matrices result in different Procrustes solutions on the corresponding Procrustes set.
\begin{figure}[!htb]
	\centering
	\includegraphics[width=.8\textwidth]{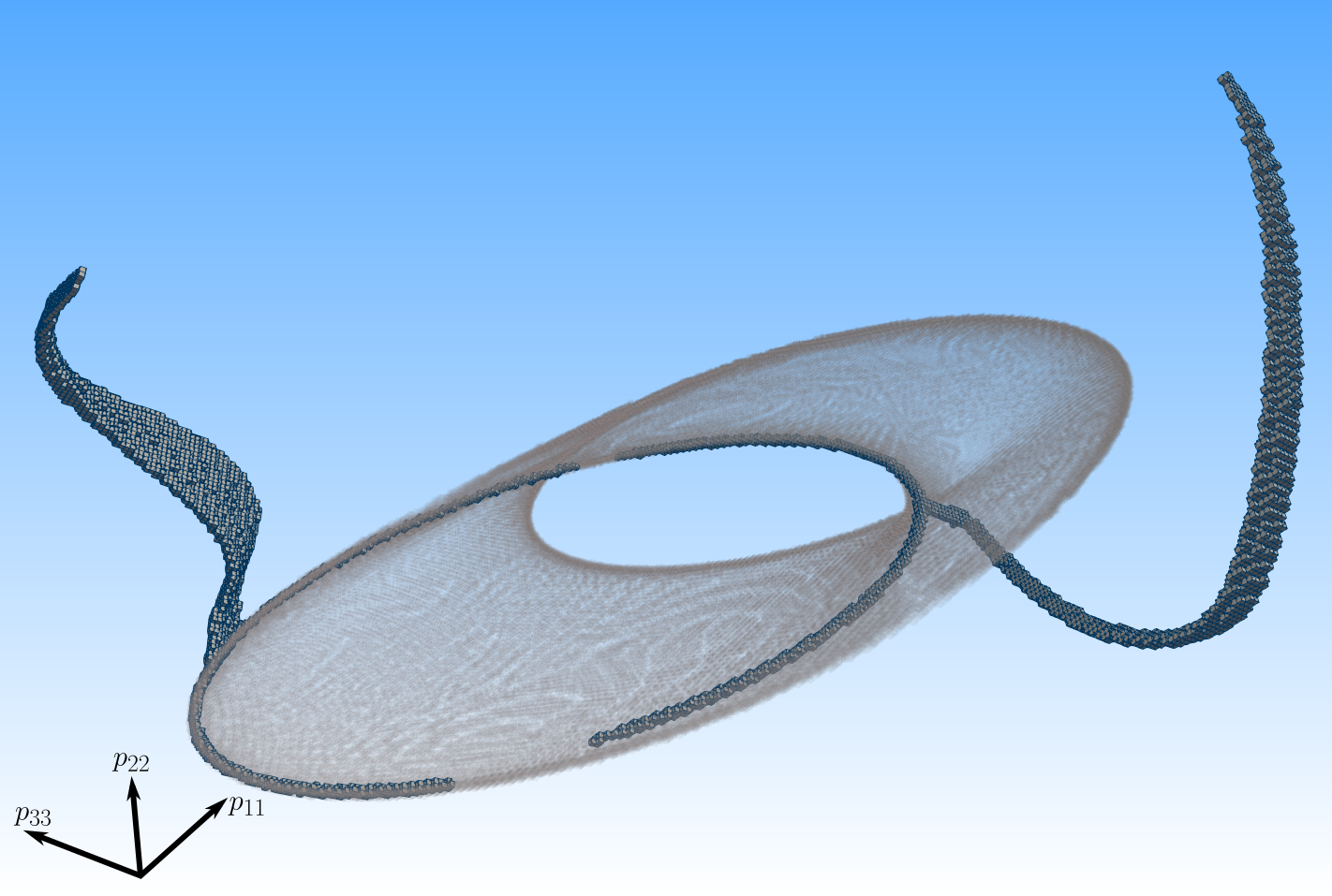}
	\caption{Three-dimensional projection of two subsets of the stable manifold. The omega-limit sets of a small neighborhood of the permutation matrices $P_1$ and $P_2$ form a half circle on their corresponding Procrustes set.}
	\label{fig:mostAttractingSets}
\end{figure}


Note that by using \eqref{eq:P_flow_noEqualityConstraint} we will always obtain a Procrustes solution, which is in general not a permutation matrix. In fact, a stability analysis shows that the optimal permutation matrix is hyperbolic with four unstable directions. We expect that those directions become stable as we enable the equality constraint in \eqref{eq:minPwithEqualityConstraints} and use the gradient flow \eqref{eq:TSP flow}. To this end, let us consider the flow \eqref{eq:TSP flow}. By using a forward difference scheme, we approximate the Jacobian in a small perturbation of the optimal permutation matrix $(\widetilde{P},\widetilde{\lambda})$ and compute the eigenvalues. The unstable directions become stable, but there are still four eigenvalues with positive (but small) real part. Note that the Lagrange multiplier, unfortunately, maps the \emph{optimal solution} of the TSP to \emph{saddle points}. Hence, although the optimal permutation matrix is still not stable, we expect the trajectories to stay in a small neighborhood of the optimal permutation matrix since the positive (unstable) eigenvalues have a small real part. Then we can extract the optimal permutation matrix.
%

Next we will numerically analyze how likely it is to obtain the optimal permutation matrix by using \eqref{eq:TSP flow}. Hence, we compute the basin of attraction (subsets of the stable manifold)
 of the optimal permutation matrix
\begin{equation*}
	\bar P = \begin{pmatrix}
	1 & 0 & 0 & 0 & 0 \\ 0 & 0 & 1 & 0 & 0 \\ 0 & 0 & 0 & 1 & 0 \\ 0 & 1 & 0 & 0 & 0 \\ 0 & 0 & 0 & 0 & 1
	\end{pmatrix}.
\end{equation*}
Again, we use an adaption of the set-oriented continuation technique by Dellnitz et al., where we integrate \eqref{eq:TSP flow} backward in time. Therefore, we only have to change the time-$\tau$-map $\Phi$ in Algorithm~\ref{alg:continuation} to the one that corresponds to the new dynamical system shown in~\eqref{eq:TSP flow}. We compute the small perturbations of the optimal permutation matrix by polar decomposition.
We set $\bQ = [-1,1]^{25}$ and choose $s = 175$ and $\tau = 10000$ as before. In Figure~\ref{fig:polar} (a)--(b) we show different three-dimensional projections of the basin of attraction of $(\widetilde{P},\widetilde{\lambda})$. The dark cells depict the stationary solutions of the gradient flow \eqref{eq:TSP flow} backward in time. There are four different stationary solutions which are also permutation matrices, two of which are shown below,
\begin{equation*}
	\bar P_1 = \begin{pmatrix}
	0 & 0 & 0 & 0 & 1 \\ 0 & 0 & 1 & 0 & 0 \\ 0 & 0 & 0 & 1 & 0 \\ 1 & 0 & 0 & 0 & 0 \\ 0 & 1 & 0 & 0 & 0
	\end{pmatrix} \quad \mbox{and}\quad \bar P_2 = \begin{pmatrix}
	1 & 0 & 0 & 0 & 0 \\ 0 & 0 & 0 & 0 & 1 \\ 0 & 0 & 0 & 1 & 0 \\ 0 & 0 & 1 & 0 & 0 \\ 0 & 1 & 0 & 0 & 0
	\end{pmatrix}.
\end{equation*}
Observe that the Procrustes sets of the gradient flow \eqref{eq:P_flow_noEqualityConstraint} are not covered by the basin of attraction. In fact, this is clear since we start in a local minimizer and we solve the gradient flow backward in time. This is equivalent to maximizing the cost function \eqref{eq:minPwithEqualityConstraints} subject to the equality constraints. Since the Procrustes solution is the global minimizer of the cost function without the constraints, it is, in general, not possible to find this solution via the gradient flow backward in time. Hence, it cannot be in the basin of attraction of the optimal solution of the TSP when using the dynamical system given by~\eqref{eq:TSP flow}.
The box-counting dimension (see, e.g., \cite{HK99}) of the covering of the basin of attraction is $d\approx 3.5$ (about $24$ million cells in the $25$-dimensional space).


%

\begin{figure}[!htb]
	\begin{minipage}{0.49\textwidth}
		\includegraphics[width = \textwidth]{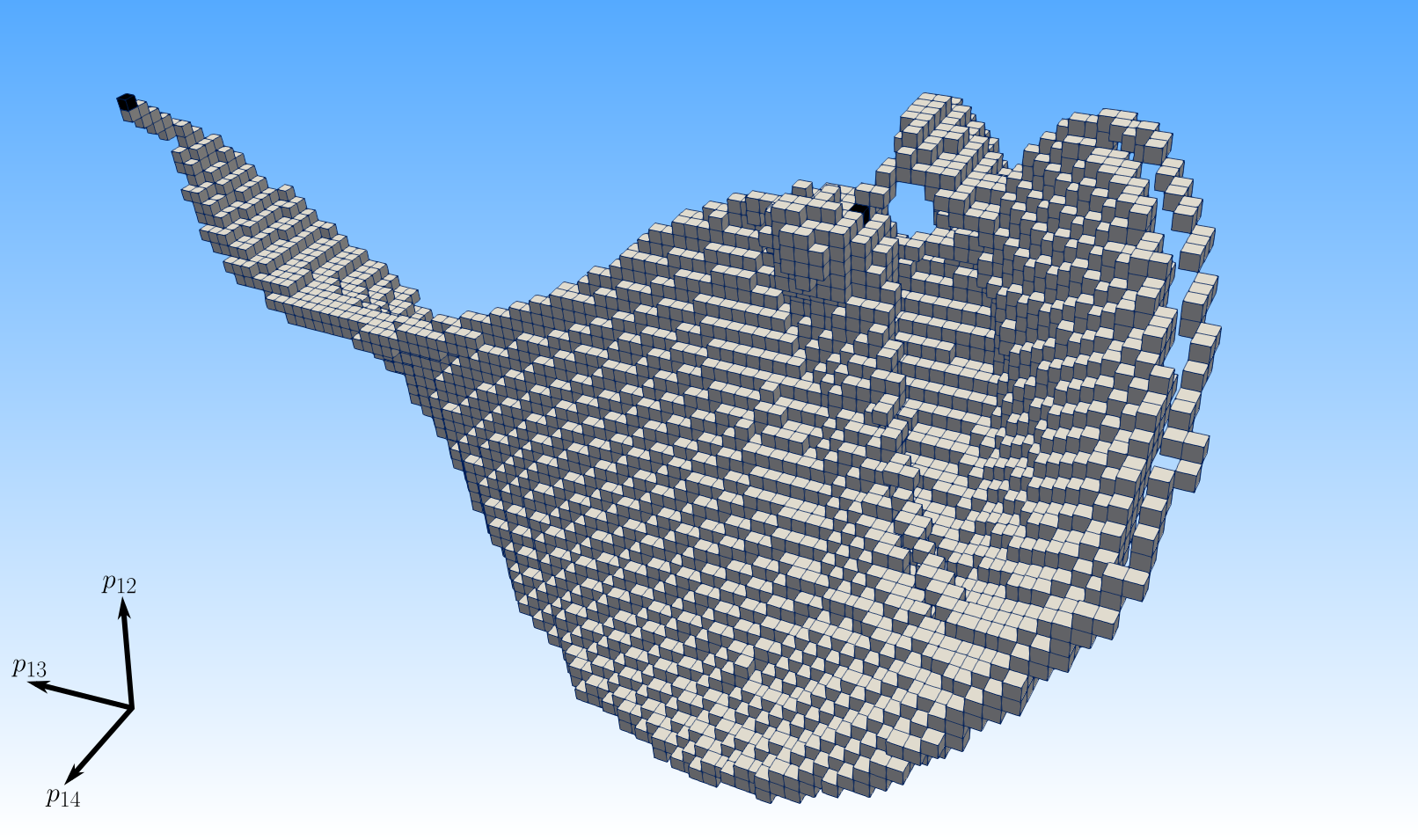}
		\centering \scriptsize{(a)}
	\end{minipage}
	\begin{minipage}{0.49\textwidth}
		\includegraphics[width = \textwidth]{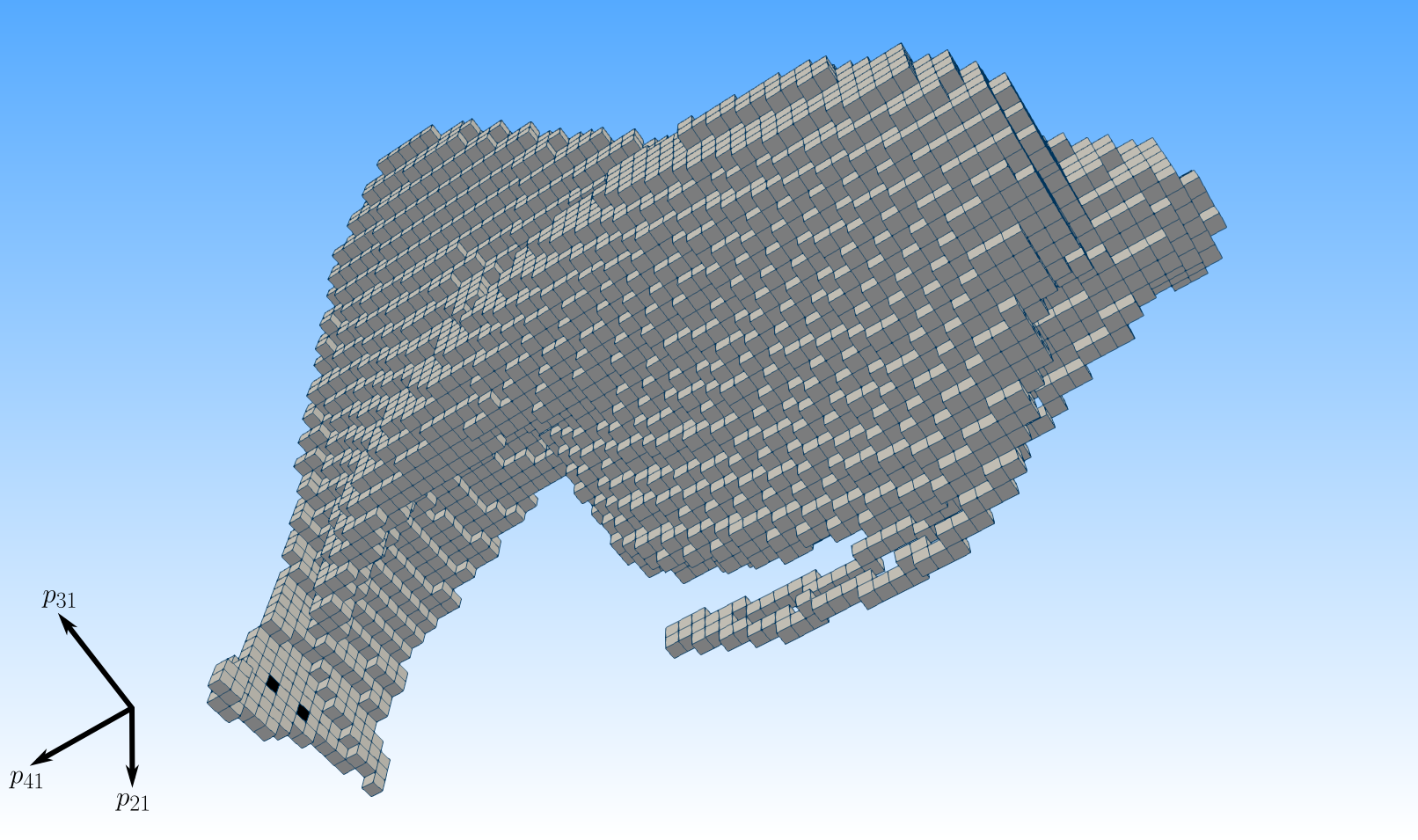}
		\centering \scriptsize{(b)}
	\end{minipage}

	\caption{(a)--(b) Three-dimensional projections of the basin of attraction of $(\widetilde P,\widetilde{\lambda})$. The dark cells depict the stationary solutions of the gradient flow \eqref{eq:TSP flow} backward in time.}
	\label{fig:polar}
\end{figure}


In an effort to address the complications due to the infinite number of Procrustes solutions, we now reformulate the dynamical systems using tour matrix representations of the solutions.


\subsection{Gradient flows for tour matrices}
\label{sec:tour_mat}
Instead of forcing the gradient flow to converge to a permutation matrix, an alternative approach is to define a cost function in such a way that the flow converges directly to a permutation of the initial tour matrix $ T $. For the symmetric case~\eqref{eq:SQAPFlow}, Brockett~\cite{Bro91} introduces a change of variables, given by $ H = P^T B P $. The resulting double bracket flow, $ \dot{H} = 2 \left[ H, \left[H, A \right] \right] $, then evolves in the space of symmetric matrices and is only quadratic in $ H $.

We now extend this to the nonsymmetric case. Assuming that the objective function $ F(P) $ can be rewritten as a function $ F(H) $, we want to derive a gradient flow for the new variable $ H $. For the TSP, again $ A = D $ and $ B = T $. Note that, as before, the $T$ matrix is replaced by $T_\text{undir}$ and $T_\text{dir}$ for symmetric and asymmetric TSP instances respectively. With the aforementioned transformation, the cost function for the relaxed TSP from Lemma~\ref{lem:relaxed QAP} can be written as
\begin{equation} \label{eq:LTSP}
    \min_{H \in \mathcal{T}_n} \tr \left( A^T H \right),
\end{equation}
where $ \mathcal{T}_n = \{ P^T T P \mid P \in \mathcal{O}_n \} $.

\begin{remark}
For directed cycle graphs, i.e., $ B = T_\text{dir} $, the non-relaxed version of \eqref{eq:LTSP} is identical to the cost of the linear assignment problem (LAP)~\cite{Cit:klus2014spectral}, with the difference that here the set of feasible solutions is constrained to a subset of $ \mathcal{P}_n $, which makes the problem NP-hard.
\end{remark}

\begin{theorem} \label{thm:H flow}
Let $ F(H) $ be a given cost function, then the gradient flow is given by
\begin{equation*}
    \dot{H} = -\left[ H, \left\{ H, F_H \right\} + \left\{ H^T, F_H^T \right\} \right].
\end{equation*}
\end{theorem}
\begin{proof}
Since $ H = P^T B P $, using \eqref{eq:GradFlow} and \eqref{eq:grad F} we obtain
\begin{equation*}
    \dot{H} = \dot{P}^T B P + P^T B \dot{P} = -\left[ H, \left\{ P, F_P \right\} \right].
\end{equation*}
Applying the chain rule, this leads to
\begin{align*}
    \pd{F}{P_{ij}} &= \tr \left( \left( \pd{F}{H} \right)^T \pd{H}{P_{ij}} \right) \\ 
                   &= \tr \left( F_H^T  \left( P^T B J^{ij} + J^{ji} B P \right) \right) \\ 
                   &= \left( B^T P F_H + B P F_H^T \right)_{ij}, 
\end{align*}
where $ J^{ij} \in \R^{n \times n} $ is a single-entry matrix~\cite{PP08}, i.e., $ (J^{ij})_{kl} = \delta_{ik}\delta_{jl} $. It follows that
\begin{equation} \label{eq:DerRel}
    F_P = B^T P F_H + B P F_H^T.
\end{equation}
Inserting this into the equation for $ \dot{H} $ concludes the proof.
\end{proof}

For the cost function $ F(H) = \tr(A^T H) $, we simply obtain $ F_H = A $. With the aid of Theorem~\ref{thm:H flow}, we can then compute the corresponding gradient flow. In addition to the cost function, a penalty function has to be used to find an admissible solution.

One possibility would be to penalize negative entries as described in Section~\ref{ssec:orthogonal_flow}. However, for the tour matrix approach we use a combination of two penalty functions. The first penalty function is given by,
\begin{equation*}
G_1(H) = \norm{H - H \circ H}_F,
\end{equation*}
which penalizes entries that are not zero or one.
Furthermore, we also use a second penalty function
\begin{equation}
	G_2(H) = \| A\cdot\mvec(H) - b \|_2^2,
\end{equation}
where $A \in \mathbb{R}^{3n\times n^2},\ b \in \mathbb{R}^{3n}$ are linear equality constraints that force the flow to converge to a matrix $H$ with row and column sums equal to two and diagonal entries equal to zero.
Given both penalty functions, we consider the following TSP as a constrained optimization problem,
\begin{equation}\label{eq:minimizeTour}
	\begin{aligned}
	 \min_{H\in \mathcal{T}_n} &\quad\mbox{tr}(A^TH),\\
		\mbox{s.t.} & \quad G_1(H) = c\,\| H-H\circ H \|_F^2 = 0,\\
		& \quad G_2(H) = c\,\| A\cdot\mvec(H) - b \|_2^2 = 0.
		\end{aligned}
\end{equation}
In order to solve \eqref{eq:minimizeTour}, we will solve the resulting system of differential equations
\begin{equation}\label{eq:tourFlow}
	\begin{aligned}
		\dot H =& -\big[ H, \big\{ H, F_H +c\,\lambda_1 G_{1,H} + c\,\lambda_2 G_{2,H} \big\}  \\
		 & + \left\{ H^T, F_H^T +c\,\lambda_1 G_{1,H}^T + c\,\lambda_2 G_{2,H}^T\right\} \big],\\
		\dot \lambda_1 =&\,c\,G_1(H),\\
		\dot \lambda_2 =&\,c\,G_2(H),
	\end{aligned}
\end{equation}
where $G_{1,H} = 2(H-H\circ H)\circ (E-2H)$ and $G_{2,H} = \mvec^{-1}\left(2(A\cdot \mbox{vec(H)}-b)^TA\right)$. Here, $ E $ denotes the matrix of ones. Again, we perform a gradient descent for the cost function and a gradient ascent for the Lagrange multipliers.
\begin{example} \label{ex:H flow}
	Let us illustrate the gradient flow with the same TSP with 10 cities as in Example~\ref{ex:P flow}. Using~\eqref{eq:tourFlow}, we obtain the results shown in Figure~\ref{fig:H flow}. We plot only the entries of the $ H $ matrix that are greater than zero. The dynamical system converges to a tour that is slightly longer than the optimal tour.
	
	\begin{figure}[htb]
		\centering
		\begin{minipage}[c]{0.32\textwidth}
			\centering
			\subfiguretitle{a)}
			\includegraphics[width=0.8\textwidth]{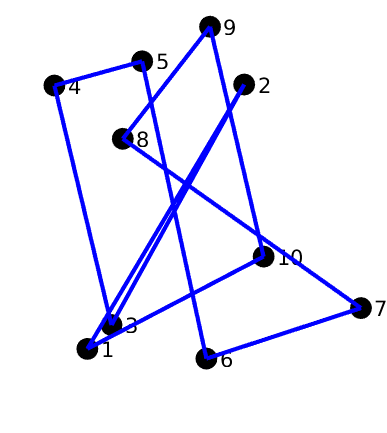}
		\end{minipage}
		\begin{minipage}[c]{0.32\textwidth}
			\centering
			\subfiguretitle{b)}
			\includegraphics[width=0.8\textwidth]{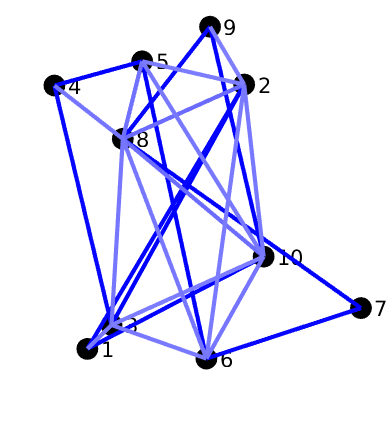}
		\end{minipage}
		\begin{minipage}[c]{0.32\textwidth}
			\centering
			\subfiguretitle{c)}
			\includegraphics[width=0.8\textwidth]{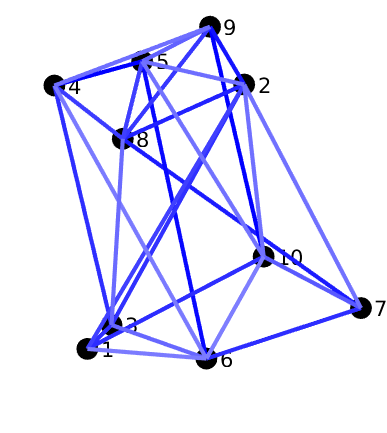}
		\end{minipage} \\[1em]
		\begin{minipage}[c]{0.32\textwidth}
			\centering
			\subfiguretitle{d)}
			\includegraphics[width=0.8\textwidth]{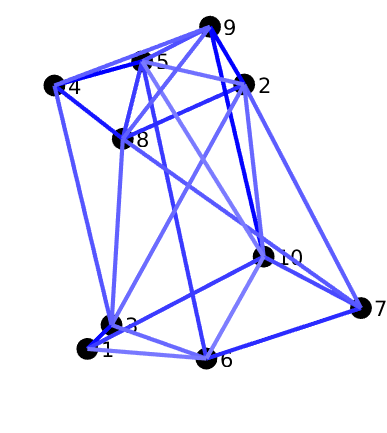}
		\end{minipage}
		\begin{minipage}[c]{0.32\textwidth}
			\centering
			\subfiguretitle{e)}
			\includegraphics[width=0.8\textwidth]{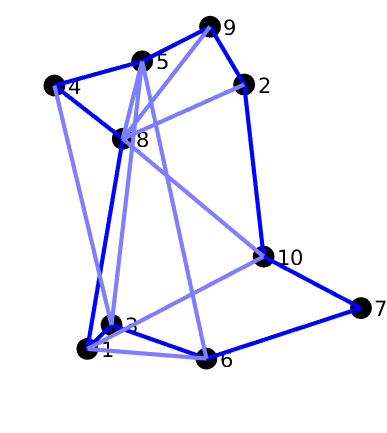}
		\end{minipage}
		\begin{minipage}[c]{0.32\textwidth}
			\centering
			\subfiguretitle{f)}
			\includegraphics[width=0.8\textwidth]{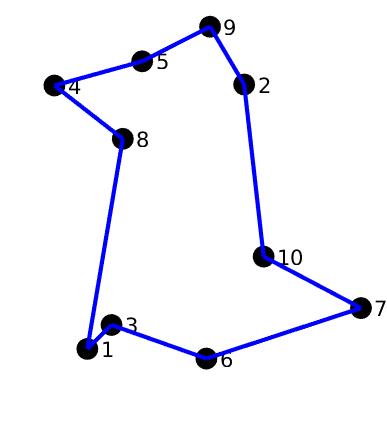}
		\end{minipage}
		\caption{Traveling salesman problem with 10 cities solved using the gradient flow~\eqref{eq:tourFlow}. a)~Initial trivial tour $ T $. b--d)~Intermediate solutions. e)~Optimal $ H $ matrix.~f)~Final solution. The initial tour matrix was transformed into a new shorter tour which is almost optimal (cf.~Figure~\ref{fig:P flow} f)).}
		\label{fig:H flow}
	\end{figure}
	
\end{example}
Analogously to Section~\ref{ssec:orthogonal_flow}, we will numerically analyze the tour matrix flow \eqref{eq:tourFlow} for the simple TSP with five cities. We first note that there exists only one Procrustes solution in $\mathcal{T}_n$, i.e., $H^* = (P^*)^T T P^*$, where $T = T_\text{undir}$ (since the distances between cities are symmetric). We note that the advantage of this formulation is that it avoids the infinite number of Procrustes solutions that were described in the previous section. In other words, even though there exist an infinite number of $P^{*}$ solutions, $H^{*}$ is unique.  Furthermore, the Procrustes solution in $\mathcal{T}_n$ is invariant under basis rotations due to repeated eigenvalues. Hence, the corresponding Procrustes set is entirely captured by $H^*$.
In order to analyze the stability of the optimal tour matrix, we take a small perturbation of the optimal tour and solve the tour matrix flow \eqref{eq:tourFlow}. This results in a matrix $\widetilde H$ with the corresponding Lagrange multipliers $\widetilde\lambda_1$ and $\widetilde\lambda_2$. Observe that $\widetilde H$ only lies in a small neighborhood of the optimal tour. Now we are in a position to compute the Jacobian in $(\widetilde H, \widetilde{\lambda}_1,\widetilde{\lambda}_2)$ using a forward difference scheme. There exist five eigenvalues with positive (but small) real part, thereby giving rise to saddle points. Thus, the optimal tour is again not stable, but we do expect that by using the tour flow \eqref{eq:tourFlow} we will stay in a small neighborhood of the optimal tour.

Finally, we again compute the basin of attraction of the optimal tour matrix for the five cities example,
\begin{equation*}
\bar H = \bar P^T  T  \bar P = \begin{pmatrix}
0 & 0 & 1 & 0 & 1 \\ 0 & 0 & 0 & 1 & 1 \\ 1 & 0 & 0 & 1 & 0 \\ 0 & 1 & 1 & 0 & 0 \\ 1 & 1 & 0 & 0 & 0
\end{pmatrix}.
\end{equation*}
Observe that a symmetric matrix $H \in \mbox{Symm}_5 = \{ A \in \R^{5 \times 5}\, \vert\, A = A^T\}$ is fully described by the lower left or upper right part of the matrix. Moreover, as described previously, the tour matrix flow \eqref{eq:tourFlow} is a matrix differential equation evolving on the manifold of $\mbox{Symm}_5$. Thus, starting with a symmetric matrix $H \in \mbox{Symm}_5$, the trajectory remains for all time in $\mbox{Symm}_5$. This allows us to choose $\bQ = [-2,2]^{15}$ and a fine partition $\bL_s$ of $\bQ$ for $s = 105$. Again, we set $\tau = 10000$ and make use of Remark~\ref{rmk:continuation_algorithm}~c). Small perturbations of the optimal tour matrix $\bar H$ are computed as follows: Let $C \in \bL_{105}$ be the cell which contains the components of the optimal tour matrix $\bar H$. A point $x \in C$ defines the components of a lower triangular matrix. By taking the lower left part of $\widetilde\mvec^{-1}(x)$, we can directly create a symmetric matrix, which is a small perturbation of the optimal tour. Observe that $\widetilde{\mvec}^{-1}$ has to be adapted accordingly.

In Figure~\ref{fig:BOATour} (a)--(b) we show different three-dimensional projections of the basin of attraction of $(\widetilde H,\widetilde{\lambda}_1,\widetilde{\lambda}_2)$. The dark cells depict the stationary solutions of \eqref{eq:tourFlow} backward in time. We note that there are five different tour matrices to which the gradient flow converges backward in time, two of the five unique matrices (each one depending on a permutation matrix with a unique cycle) are shown below,
\begin{equation*}
\bar H_1 = \begin{pmatrix}
0 & 0 & 0 & 1 & 1 \\ 0 & 0 & 1 & 0 & 1 \\ 0 & 1 & 0 & 1 & 0 \\ 1 & 0 & 1 & 0 & 0 \\ 1 & 1 & 0 & 0 & 0
\end{pmatrix} \quad \mbox{and}\quad \bar H_2 = \begin{pmatrix}
0 & 0 & 0 & 1 & 1 \\ 0 & 0 & 1 & 1 & 0 \\ 0 & 1 & 0 & 0 & 1 \\ 1 & 1 & 0 & 0 & 0 \\ 1 & 0 & 1 & 0 & 0
\end{pmatrix}.
\end{equation*}
\begin{figure}[!htb]
	\begin{minipage}{0.49\textwidth}
		\includegraphics[width = \textwidth]{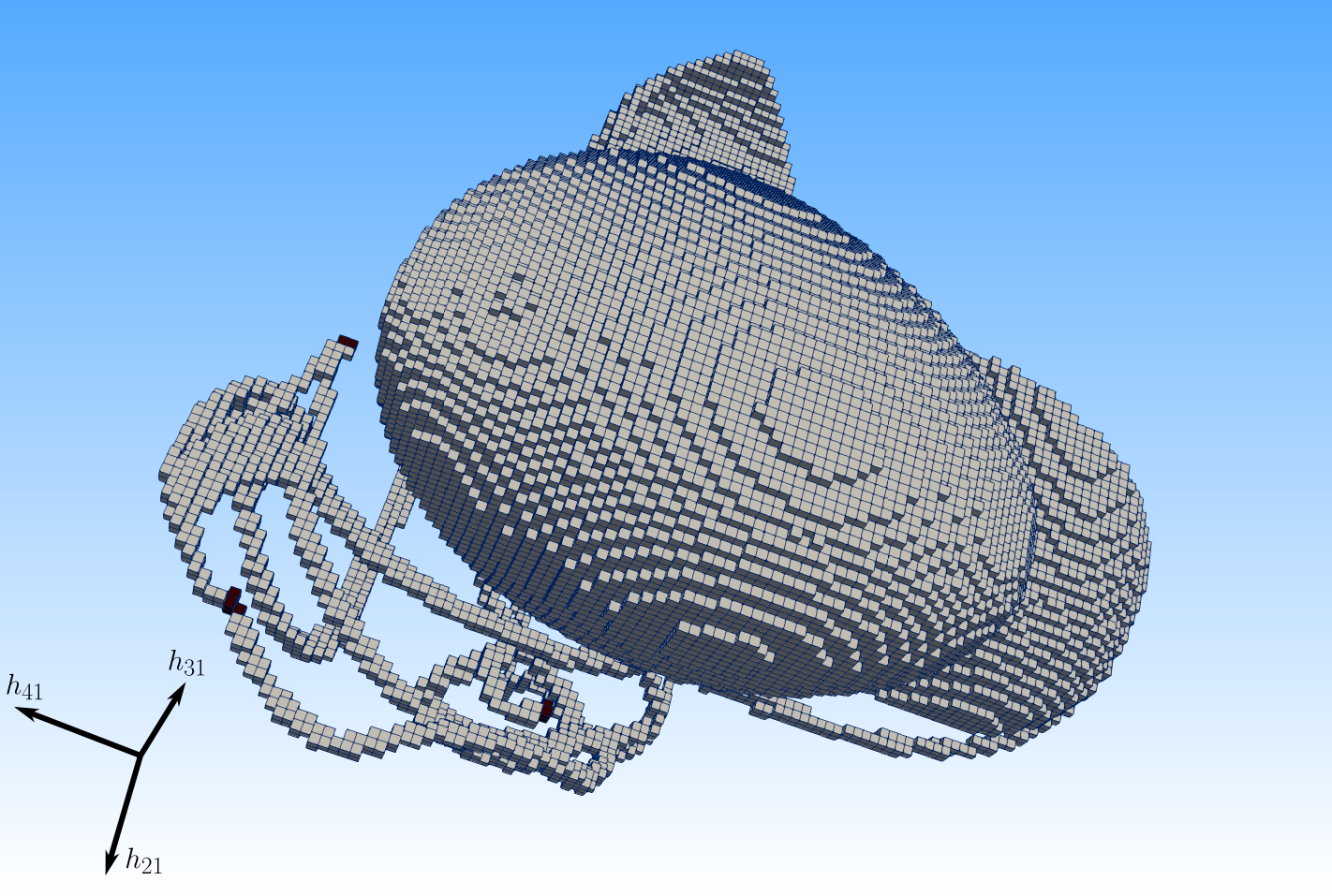}
		\centering \scriptsize{(a)}
	\end{minipage}
	\begin{minipage}{0.49\textwidth}
		\includegraphics[width = \textwidth]{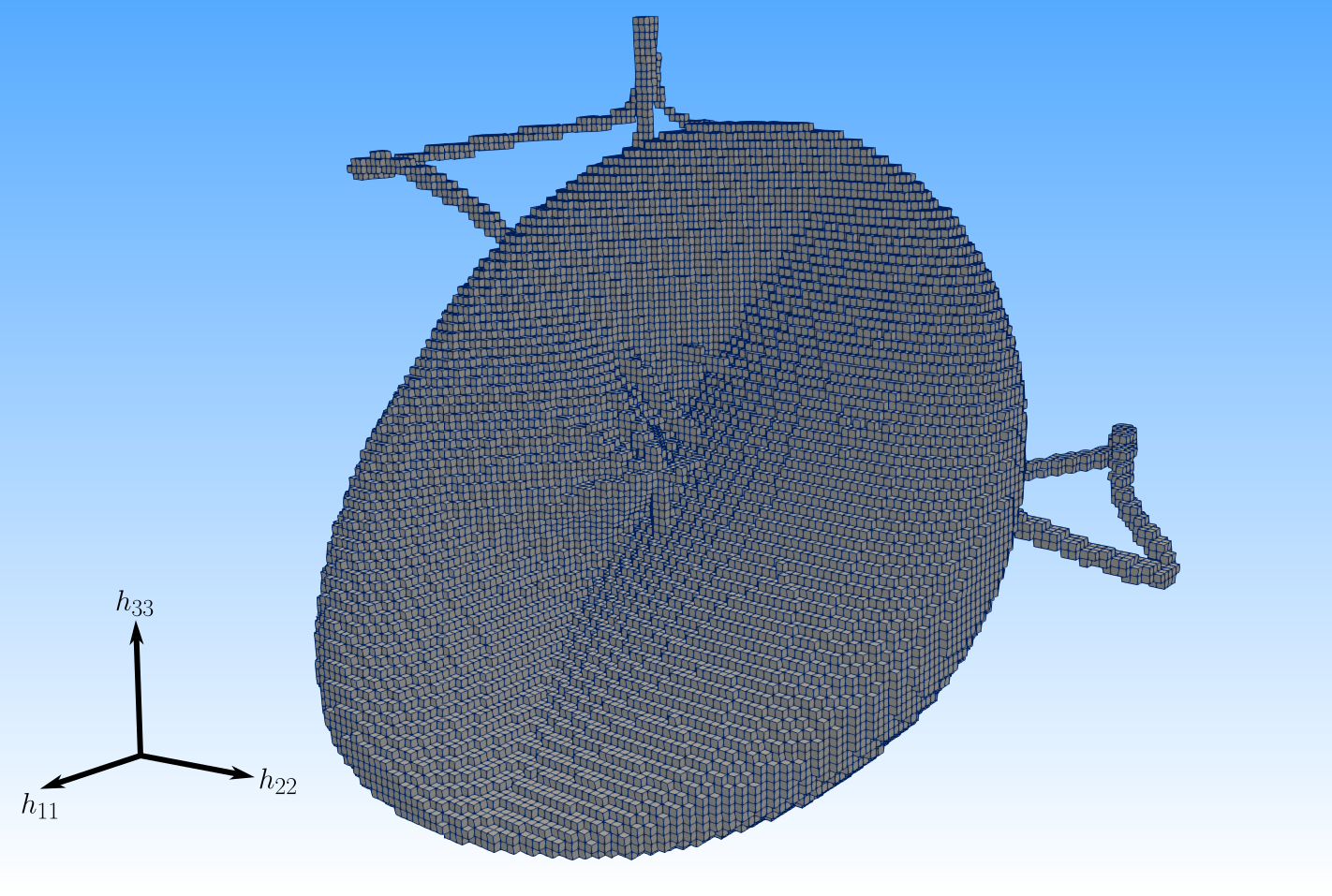}
		\centering \scriptsize{(b)}
	\end{minipage}
	
	\caption{(a)--(b) Three-dimensional projections of the basin of attraction of $(\widetilde H,\widetilde{\lambda}_1,\widetilde{\lambda}_2)$. The dark cells depict the stationary solutions of the gradient flow backward in time.}
	\label{fig:BOATour}
\end{figure}

Furthermore, the box-counting dimension of the basin of attraction is about ${d \approx 3.28}$.

Although the flows described above are interesting and display the complexity and challenges presented by NP-hard problems from a dynamical system lens, we find that the solutions computed by the above methods are not competitive when compared to state-of-the-art methods. Our goal in the next section is to construct new variants of existing state-of-the-art methods that are inspired by our work above.
\section{Procrustes-based Lin--Kernighan heuristic}
\label{sec:Procrustes}

In this section, we will propose a method to compute candidate sets based on the relaxed problem~\eqref{eq:RelTSPCost} or \eqref{eq:LTSP}, respectively. The solution, which is given by the solution of the Procrustes problem (see section~\ref{sec:twosidedproc}), can be computed analytically. Note that the solutions computed using this approach are optimal solutions of flows (for $P$ and $H$) described in previous sections.

Thus, the optimal solution which minimizes the relaxed TSP cost function~\eqref{eq:RelTSPCost} can be obtained using the solution of the Procrustes problem described in Section~\ref{sec:twosidedproc}. Namely, the solution is given in terms of matrix $ P^* = V_T V_D^T $ \cite{AW00,HRW92}, where the eigenvectors in the two matrices are sorted with respect to increasing eigenvalues of $T$ and decreasing eigenvalues of $D$ or vice versa. Define $ T^* = P^{*T} T P^* = V_D \Lambda_T V_D^T $ to be the solution of the two-sided orthogonal Procrustes problem or the minimum of the tour flow. Note that $ T^* $ is also the optimal solution of the gradient flow \eqref{eq:tourFlow} without both equality constraints, i.e.,
\begin{align*}
	\dot H =& -\big[ H, \big\{ H, F_H\big\}+\big\{ H^T, F_H^T\big\}\big].
\end{align*}
Roughly speaking, $ T^* $ can be interpreted as a \emph{continuous} solution of the relaxed TSP where the entry $ t_{ij}^* $ describes the strength of edge $ (i, j) $. We would like to now use the entries of $T^*$ to help inform the Lin--Kernighan heuristic. In particular, we use the solution of the Procrustes problem to limit the search space and to improve the efficiency of local heuristics. The aim is to bias $ k $-opt moves in a manner such that edges with high edge strengths are included with high probability. We compute candidate sets based on the entries $ t^*_{ij} $ of the matrix $ T^* $ and call these matrix values \emph{$ P $-nearness}. In our proposed approach, for each city, we pick the cities with the largest entries $ t_{ij}^* $.

\begin{example}
Let us consider again the TSP from Examples~\ref{ex:P flow} and~\ref{ex:H flow}. Figure~\ref{fig:OptPerm and OptRot} shows the difference between the optimal solution of the TSP and the optimal solution of the Procrustes problem. We use a linear interpolation between blue (large $ t_{ij}^* $ value) and white (small $ t_{ij}^* $ value). Note that this is the same matrix as in Figure~\ref{fig:H flow}d, with the difference that we are plotting a few more edges here to illustrate $ P $-nearness. Clearly, some edges of the optimal tour are already visible in Figure~\ref{fig:OptPerm and OptRot}b, for example $ (2, 9) $ and $ (5,9) $, while other edges such as $ (3, 10) $ or $ (2, 5) $ have a much lower weight. For city $ 10 $, different choices exist, $ (2, 10) $, $ (6, 10) $, $(7,10)$, and $ (8, 10) $, for instance, have a high probability of being part of the shortest tour.

\begin{figure}[htb]
    \centering
    \begin{minipage}[t]{0.32\textwidth}
        \centering
        \subfiguretitle{a)}
        \includegraphics[width=0.8\textwidth]{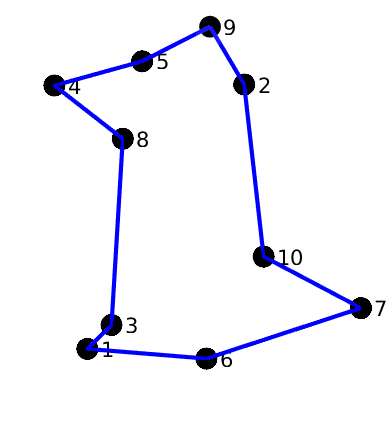}
    \end{minipage}
    \begin{minipage}[t]{0.32\textwidth}
        \centering
        \subfiguretitle{b)}
        \includegraphics[width=0.8\textwidth]{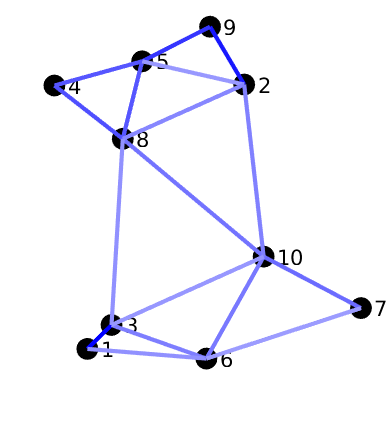}
    \end{minipage}
    \caption{Solution of a TSP with 10 cities. a) Optimal permutation matrix. b) Optimal orthogonal matrix.}
    \label{fig:OptPerm and OptRot}
\end{figure}

\end{example}

In Figure~\ref{fig:distance_vs_procrustes}, we show the edges computed using the Procrustes solution. In particular, for two random TSP instances ($50$ city and $100$ city examples), we show the shortest edges in the left-most column and the edges from the Procrustes solution in the right-most column. The optimal tour is plotted in the middle column. It is evident from the figure that the Procrustes solution $T^{*}$ tends to capture most of the edges in the optimal tour. Note that, while the $ \alpha $-nearness values can be computed in $ \mathrm{O}(n^2) $~\cite{HK70,Hel98,Cit:keld2}, the computation of the $ P $-nearness values is $ \mathrm{O}(n^3)$.

\begin{figure}[htb]
    \centering
    \includegraphics[width=0.95\textwidth]{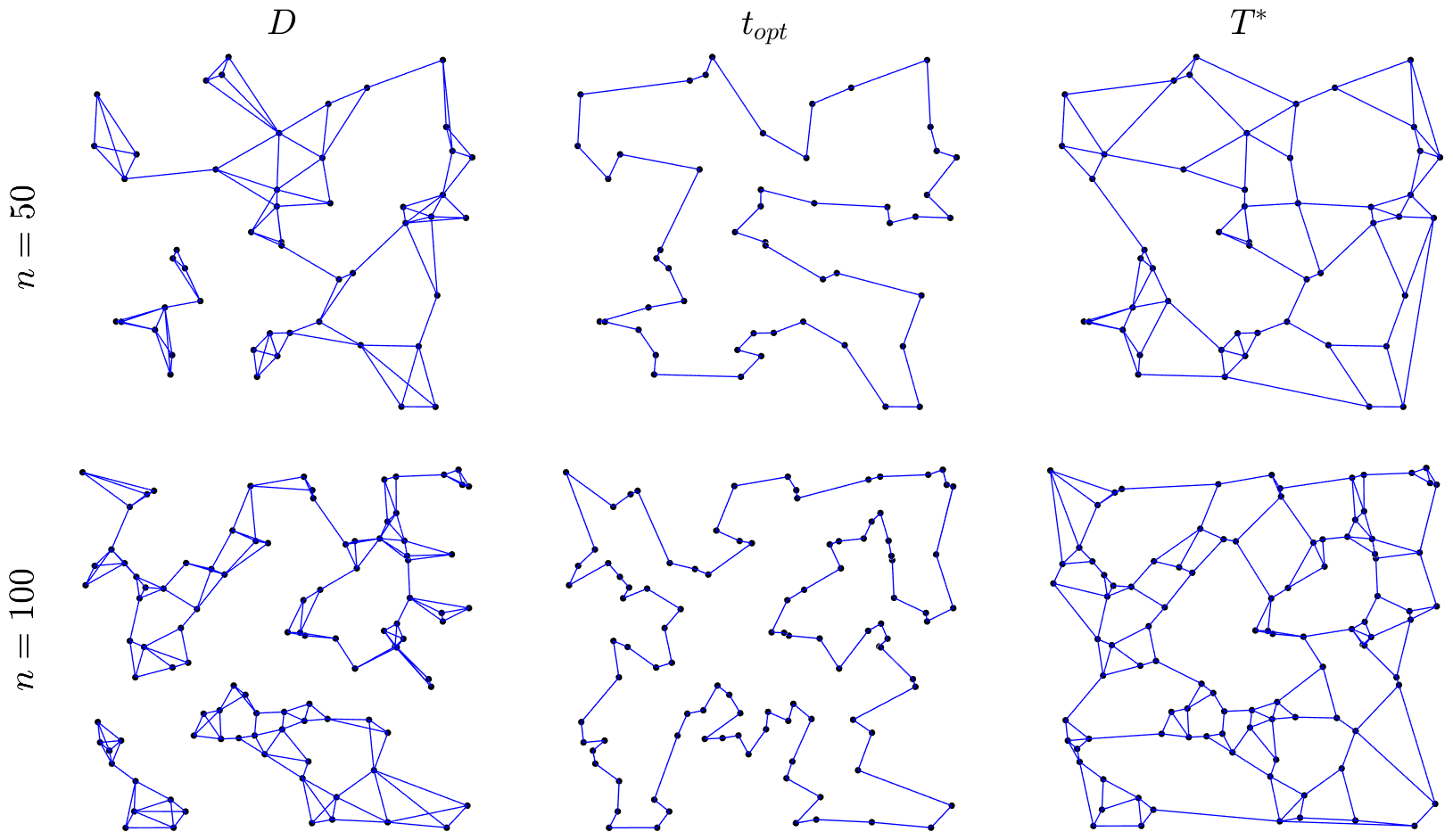}
    \caption{Illustration of $ P $-nearness for random TSP instances of size $50$ and $100$. The left column contains the edges with shortest distance, the center column has the optimal tour for the instances, and the right column contains the edges with the highest $ P $-nearness values for each city. For each city, we plotted the three edges with the highest nearness values.}
    \label{fig:distance_vs_procrustes}
\end{figure}

\subsection{Improving the Procrustes solution}

Analogous to the $\alpha$-nearness approach, we now describe our methodology to obtain better $P$-nearness values than those computed from the solution of the Procrustes problem alone. As described in section~\ref{ssec:subgradientOptimization}, the $\alpha$-nearness values were improved using subgradient optimization. In the $P$-nearness setting,  the principal idea is to construct a homotopy between the original TSP distance matrix $D$ and the solution of the Procrustes solution. Intuitively, one desires the candidate sets to include `several' short edges and a `few' long edges. Note that simply picking the shortest edges from $D$ gives rise to greedy solutions that are usually not competitive since they require the addition of long edges to complete the tours~\cite{Cit:cook}.

We find that the Procrustes solution $ T^* $ tends to select too many long edges (as shown in Figure~\ref{fig:homotopy}a). If we use the entries of $ T^* $ to bias the Lin--Kernighan heuristic, the solutions are found to be close, but less competitive than the standard approach. An effort to reduce the number of long edges is equivalent to making the computed solution ``greedy'' by picking edges based on the distance matrix $D$. Thus, we construct a homotopy $ \tilde{H} $ of the form,
\begin{equation*}
    \tilde{H} = T^{*} - \lambda D.
\end{equation*}
The candidate sets for varying $\lambda$ are shown for an example TSP instance in Figure~\ref{fig:homotopy}. To find the optimal $\lambda$ we use ideas from graph clustering~\cite{Cit:sahai_hearing}, where one computes the existence of disconnected clusters in the graph. In particular, we increase $\lambda$ until the graph of candidate sets is almost disconnected (separated into clusters). This optimal $\lambda$ is found by either marching in $\lambda$ or using a bisection approach.

There are multiple ways that one can compute the connectedness of a graph. In particular, one can use a depth-first search based approach~\cite{Cit:intro_algorithms} or perform computations on the graph Laplacian~\cite{Cit:chung,Cit:Fiedler1,Cit:Fiedler2}. The rank of the graph Laplacian matrix is related to the number of connected components in the graph~\cite{Cit:chung}. In our work, we pick the graph Laplacian approach for computing connected components in the graph (by looking at the multiplicity of the zero eigenvalue). Note that these computations can also be performed in the distributed setting~\cite{sahai2010wave,Cit:sahai_hearing}. If varying $\lambda$ does not give rise to a disconnected graph, we set $\lambda=1$. Alternatively, one can use the $ D $ matrix (in place of $T^* -\lambda D$) to bias the $k$-opt moves in the Lin--Kernighan heuristic. If the candidate sets based on distance only are connected, this typically implies that the $k$-opt moves converge quickly to the shortest tour.

\begin{figure}[htb]
    \centering
    \begin{minipage}[t]{0.3\textwidth}
        \centering
        \subfiguretitle{a) $\lambda=0.0$}
        \includegraphics[width=\textwidth]{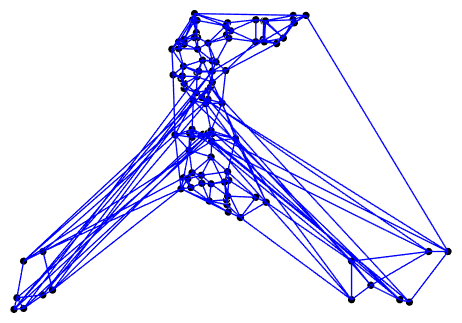}
    \end{minipage}
    \begin{minipage}[t]{0.3\textwidth}
        \centering
        \subfiguretitle{b) $\lambda=0.5$}
        \includegraphics[width=\textwidth]{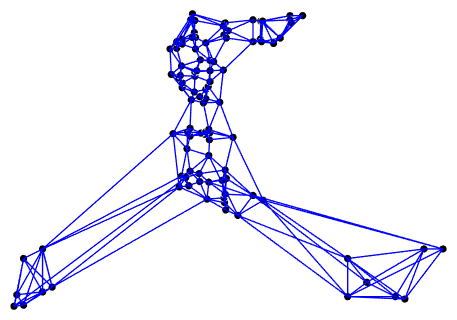}
    \end{minipage}
    \begin{minipage}[t]{0.3\textwidth}
        \centering
        \subfiguretitle{c) $\lambda=1.0$}
        \includegraphics[width=\textwidth]{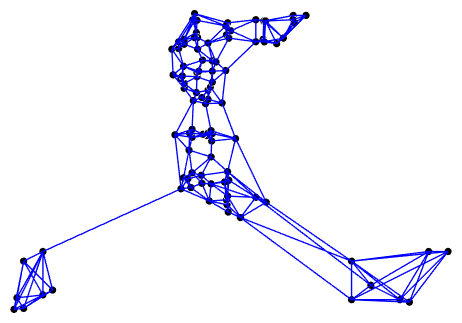}
    \end{minipage}
    \caption{Homotopy of $T^{*} -\lambda D$.}
    \label{fig:homotopy}
\end{figure}
\section{Results}
\label{sec:Results}

To compare different candidate sets or methods to bias $ k $-opt, Helsgaun computes the ``average rank'' of the edges which form the shortest tour~\cite{Hel98}. The ranking is essentially an ordering on set of edges for each node. This ordering captures the ``likelihood'' of an edge being in the optimal tour and is typically computed using the $\alpha$-nearness values. In our proposed approach, the $\alpha$-nearness values are replaced by $P$-nearness values. Thus, the optimal average rank is $ 1.5 $, all edges belonging to the shortest tour have either rank $ 1 $ or $ 2 $. We found that the average rank is in general not a good metric for the quality of the nearness values or candidate sets. Although the average rank of the Procrustes solution is typically much higher than the average rank of the $ \alpha $-nearness values, $ k $-opt often converges faster to the shortest tour.

In order to compare $ \alpha $-nearness and $ P $-nearness, we compute tours using Helsgaun's LKH package. For each LKH run, we generate the candidate sets based on the $ \alpha $-nearness and $ P $-nearness values. Starting from initial tours computed using $\alpha$-nearness and $P$-nearness, respectively, we compare the resulting tour lengths after a fixed number of $ k $-opt steps.

In Table~\ref{tab:tsplib}, we compare $ 22 $ well-known instances of the TSP from the TSPLIB database~\cite{Cit:tsplib}. The size of the candidate sets in these computations is fixed, we compute $ 5 $ candidates for each city using $ \alpha $-nearness or $ P $-nearness values, respectively. Starting from a random initial tour that is generated from the respective candidate sets, we perform a fixed number of $ 8 n $ $k$-opt moves, where $ n $ is the number of cities. We find that in this setting, the $ P $-nearness based approach typically converges faster than $ \alpha $-nearness. For example, after $ 8 n $ steps, $ P $-nearness based LKH converges to lower cost values in $18$ of the instances when compared to $\alpha$-nearness based LKH. Moreover, we ran $50$ random TSP instances of size $1000$  (cities) and found that $P$-nearness had lower tour costs after a fixed number of $k$-opt moves in $31$ of the instances, hence resulting in better solutions in $62\%$ of the instances. Note that if we run both, $\alpha$-nearness and $P$-nearness based LKH, to convergence, both methods compute the best known optimal tours in these instances. Since the initial tours are constructed using the candidate sets, the starting costs may occasionally differ slightly when comparing $\alpha$-nearness with $P$-nearness.

We do not present runtime results comparing the two algorithms since our prototype code was implemented in MATLAB and is consequently unable to compete with LKH (implemented in C) in speed. Our MATLAB implementation also limits the size of the TSP instances that we can handle. In future work, we intend to re-implement our algorithm in C for greater scalability and performance. Moreover, our approach will have higher computational cost than $\alpha$-nearness based methods due to the $O(n^3)$ eigenvector computations. However, given that the $P$-nearness requires fewer iterations of the Lin--Kernighan heuristic, we conjecture that by combining our approach with fast spectral methods~\cite{cit:schafer2017owhadi}, one can construct a highly competitive TSP approach.

\begin{table}[htbp]
	\caption{Comparison of $\alpha$-nearness and $P$-nearness based Lin--Kernighan heuristic on TSPLIB instances. The size of the candidate sets is set to $5$ per city and we stop the computations after $ 8 n $ $k$-opt moves in LKH, where $n$ is the number of cities. Out of the $ 22 $ TSPLIB instances, $P$-nearness computes a better solution in $18$ of the cases.}
	\centering
	\begin{tabular}{|l|r|r|r|}
		\hline
		\textbf{TSP} & $\alpha$-nearness & $P$-nearness & improvement \\
		\hline
		d198    &  16540 &  16465 &  0.45 \%  \\ \hline
		pcb442  &  50785 &  50832 & -0.09 \% \\ \hline
		d493    &  36028 &  35023 &  2.79 \% \\ \hline
		u574    &  36984 &  36926 &  0.16 \% \\ \hline
		rat575  &   6796 &   6790 &  0.09 \% \\ \hline
		p654    &  35716 &  37039 & -3.70 \% \\ \hline
		d657    &  49504 &  49158 &  0.70 \% \\ \hline
		u724    &  42295 &  41904 &  0.92 \% \\ \hline
		rat783  &   9054 &   8810 &  2.69 \% \\ \hline
		pr1002  & 261797 & 259810 &  0.76 \% \\ \hline
		u1060   & 224510 & 224552 & -0.20 \% \\ \hline
		vm1084  & 244411 & 242573 &  0.75 \% \\ \hline
		pcb1173 &  56934 &  56915 &  0.03 \% \\ \hline
		d1291   &  53357 &  51610 &  3.27 \% \\ \hline
		rl1323  & 279810 & 275904 &  1.40 \% \\ \hline
		nrw1379 & 141510 &  67035 & 52.63 \%\\ \hline
		fl1400  &  21319 &  22775 & -6.83 \% \\ \hline
		u1432   & 153213 & 153054 &  0.10 \% \\ \hline
		fl1577  &  28217 &  24357 & 13.68 \% \\ \hline
		d1655   &  95532 &  64837 & 32.13 \% \\ \hline
		u1817   &  58351 &  58213 &  0.24 \% \\ \hline
		rl1889  & 345475 & 340271 &  1.51 \% \\ \hline
	\end{tabular} 	
	\label{tab:tsplib}
\end{table}

\begin{figure}[htb]
    \centering
    \begin{minipage}[t]{0.45\textwidth}
        \centering
        \subfiguretitle{a)}
        \includegraphics[width=0.8\textwidth]{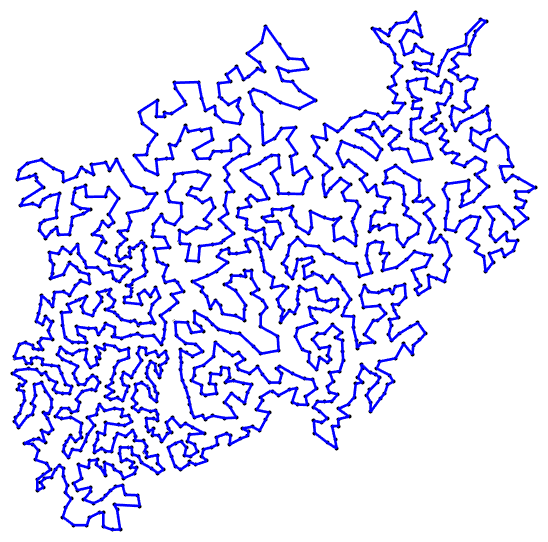}
    \end{minipage}
    \begin{minipage}[t]{0.45\textwidth}
        \centering
        \subfiguretitle{b)}
        \includegraphics[width=\textwidth]{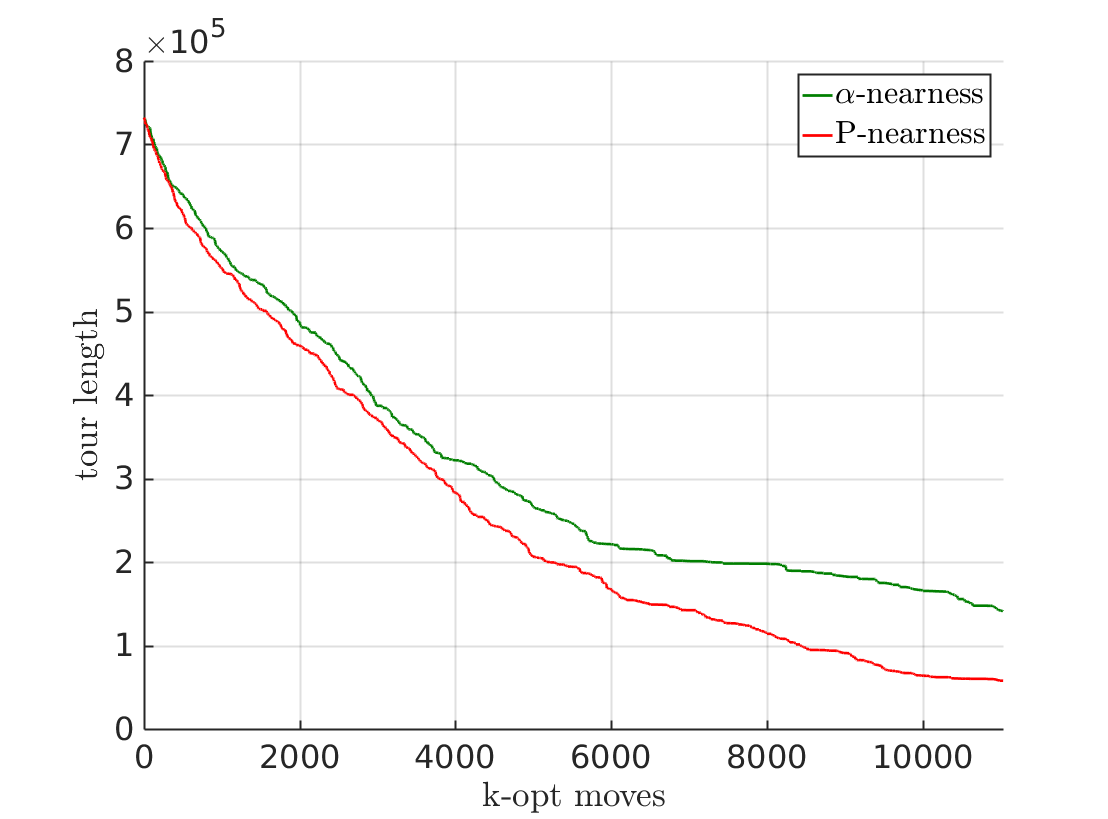}
    \end{minipage} \\[1em]
    \begin{minipage}[t]{0.45\textwidth}
        \centering
        \subfiguretitle{c)}
        \includegraphics[width=0.9\textwidth]{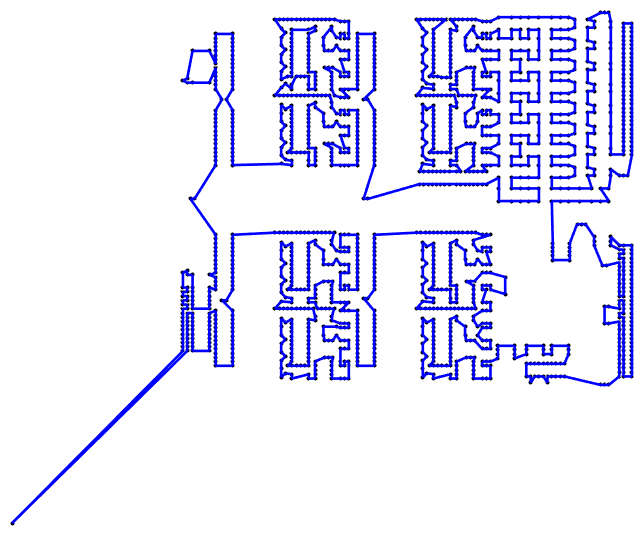}
    \end{minipage}
    \begin{minipage}[t]{0.45\textwidth}
        \centering
        \subfiguretitle{d)}
        \includegraphics[width=\textwidth]{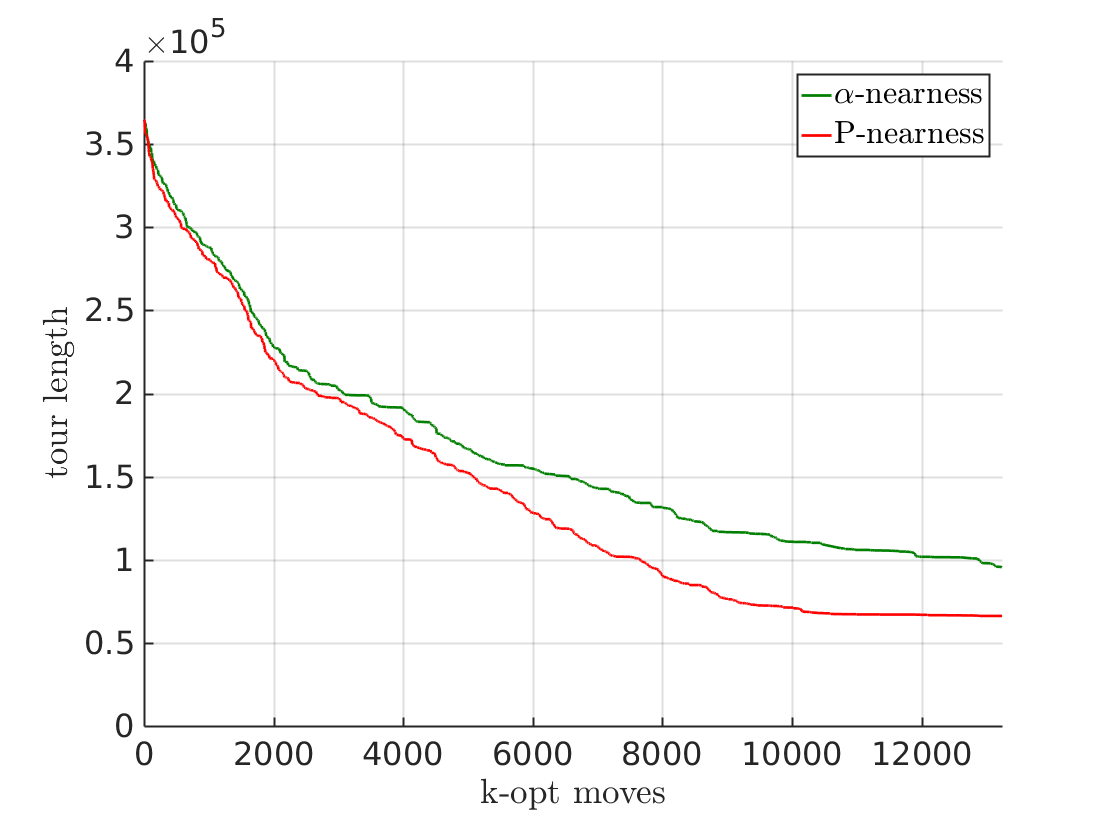}
    \end{minipage}
    \caption{TSPLIB instances and tour lengths as a function of the number of $k$-opt moves. a--b) nrw1379. c--d) d1655.}
    \label{fig:k-opt vs. tour length}
\end{figure}

\section{Conclusion and future work}
\label{sec:Conclusion}

In this work, we explored the use of continuous relaxations and dynamical systems theory for constructing algorithms for the TSP. Our approach aimed to exploit the observation that the solution of the TSP can be represented as a permutation matrix which lies on the manifold of orthogonal matrices. In the first part of this manuscript, we constructed a dynamical system on the manifold of orthogonal matrices that converges to solutions of the TSP. We also explored the construction of gradient flows for tour matrices in Section~\ref{sec:tour_mat}. We found that although the dynamical systems approach is elegant and sheds light on the structure and complexity of NP-hard problems, it often converges to local optima. We also used homotopy continuation methods to compute subsets of the stable manifolds of the optimal solutions.

Inspired by the dynamical systems approach, we then exploited a Procrustes based approximation that computes an orthogonal matrix that minimizes the TSP cost. Our approach was based on the computation of the solution of the two-sided orthogonal Procrustes problem which is based on the eigendecomposition of the corresponding tour and distance matrices of the TSP instance. We then constructed a homotopy of the Procrustes solution with the distance matrix that is then used to bias the popular Lin--Kernighan heuristic.

In certain TSP instances, the candidate sets constructed from the homotopy are found to give faster convergence than minimum spanning tree ($1$-tree) based approaches. Our algorithm was implemented in the LKH software framework and demonstrated on multiple TSPLIB and random TSP instances.

Future work includes the testing of the Procrustes approach on larger instances of the TSP by exploiting parallel eigenvector computation packages~\cite{cit:scalapack}. On the theoretical side, we aim to pursue the generalization of our proposed approach to the quadratic assignment problem (QAP). The aim is to develop efficient heuristics utilizing the results of the Procrustes problem and dynamical systems theory for solving strongly NP-hard problems. We are also exploring the use of subgradient optimization for improving the Procrustes solution which is expected to provide faster convergence rates for the TSP and related optimization problems. Additionally, we are exploring the use of fast spectral methods~\cite{cit:schafer2017owhadi} for accelerating the candidate set computations. Moreover, we hope that this work increases interest in the area at the intersection of dynamical systems theory and combinatorial optimization. There appear to be deep connections between the two areas~\cite{cit:wibisono} that may enable the construction of new optimization algorithms for a wide class of optimization problems.

\section*{Acknowledgments} \noindent The authors thank Prof.~Keld Helsgaun for discussions related to the Lin--Kernighan heuristic and his software and also Dr. Mirko Hessel-von Molo and Steffen Ridderbusch for discussions related to the approach. This material is based upon work supported by the Defense Advanced Research Projects Agency (DARPA) and Space and Naval Warfare Systems Center, Pacific (SSC Pacific) under Contract No. N6600118C4031.

\bibliographystyle{alpha}
\bibliography{TSP}

\newcommand{\etalchar}[1]{$^{#1}$}
\begin{thebibliography}{AAM{\etalchar{+}}00}

\bibitem[AAM{\etalchar{+}}00]{Cit:geneTSP}
R.~Agarwala, D.L. Applegate, D.~Maglott, G.D. Schuler, and A.A. Sch{\"a}ffer.
\newblock A fast and scalable radiation hybrid map construction and integration
  strategy.
\newblock {\em Genome Research}, 10(3):350--364, 2000.

\bibitem[ABCC98]{Cit:cuttingplane2}
D.~Applegate, R.~Bixby, W.~Cook, and V.~Chv{\'a}tal.
\newblock {\em On the solution of traveling salesman problems}.
\newblock Rheinische Friedrich-Wilhelms-Universit{\"a}t Bonn, 1998.

\bibitem[AW00]{AW00}
K.~Anstreicher and H.~Wolkowicz.
\newblock On {L}agrangian relaxation of quadratic matrix constraints.
\newblock {\em SIAM Journal on Matrix Analysis and Applications}, 22(1):41--55,
  2000.

\bibitem[BCC{\etalchar{+}}97]{cit:scalapack}
L~Susan Blackford, Jaeyoung Choi, Andy Cleary, Eduardo D'Azevedo, James Demmel,
  Inderjit Dhillon, Jack Dongarra, Sven Hammarling, Greg Henry, Antoine
  Petitet, et~al.
\newblock {\em ScaLAPACK users' guide}.
\newblock SIAM, 1997.

\bibitem[BcPP98]{BCPP98}
R.~E. Burkard, E.~\c{C}ela, P.~M. Pardalos, and L.~S. Pitsoulis.
\newblock The quadratic assignment problem.
\newblock In P.~M. Pardalos and D.-Z. Du, editors, {\em Handbook of
  Combinatorial Optimization}, pages 1713--1809. Springer, 1998.

\bibitem[BHvM09]{cit:sat_book}
Armin Biere, Marijn Heule, and Hans van Maaren.
\newblock {\em Handbook of satisfiability}, volume 185.
\newblock IOS press, 2009.

\bibitem[Bro89]{Bro89}
R.~W. Brockett.
\newblock Least squares matching problems.
\newblock {\em Linear Algebra and its Applications}, 122--124:761--777, 1989.

\bibitem[Bro91]{Bro91}
R.~W. Brockett.
\newblock Dynamical systems that sort lists, diagonalize matrices and solve
  linear programming problems.
\newblock {\em Linear Algebra and Its Applications}, 146:79--91, 1991.

\bibitem[Car97]{Cit:telescope}
S.~Carlson.
\newblock Algorithm of the gods.
\newblock {\em Scientific American}, 276:121--124, 1997.

\bibitem[Chr76]{Cit:christofides}
N.~Christofides.
\newblock Worst-case analysis of a new heuristic for the travelling salesman
  problem.
\newblock Technical report, DTIC Document, 1976.

\bibitem[Chu97]{Cit:chung}
F.~Chung.
\newblock {\em Spectral graph theory}, volume~92.
\newblock American Mathematical Soc., 1997.

\bibitem[Coo11]{Cit:cook}
W.J. Cook.
\newblock {\em In Pursuit of the Traveling Salesman: Mathematics at the Limits
  of Computation}.
\newblock Princeton University Press, 2011.

\bibitem[Cor09]{Cit:intro_algorithms}
T.H. Cormen.
\newblock {\em Introduction to algorithms}.
\newblock MIT press, 2009.

\bibitem[DFJ01]{DFJ01}
M.~Dellnitz, G.~Froyland, and O.~Junge.
\newblock {The algorithms behind GAIO - Set oriented numerical methods for
  dynamical systems}.
\newblock In {\em Ergodic theory, analysis, and efficient simulation of
  dynamical systems}, pages 145--174. Springer, 2001.

\bibitem[DG97]{Cit:Ant}
M.~Dorigo and L.M. Gambardella.
\newblock Ant colony system: A cooperative learning approach to the traveling
  salesman problem.
\newblock {\em Evolutionary Computation, IEEE Transactions on}, 1(1):53--66,
  1997.

\bibitem[DH96]{DH96}
M.~Dellnitz and A.~Hohmann.
\newblock The {C}omputation of {U}nstable {M}anifolds using {S}ubdivision and
  {C}ontinuation.
\newblock In {\em Nonlinear {D}ynamical {S}ystems and {C}haos}, pages 449--459.
  Springer, 1996.

\bibitem[DH97]{DH97}
M.~Dellnitz and A.~Hohmann.
\newblock A subdivision algorithm for the computation of unstable manifolds and
  global attractors.
\newblock {\em Numerische Mathematik}, 75:293--317, 1997.

\bibitem[EAS98]{EAS98}
A.~Edelman, T.~A. Arias, and S.~T. Smith.
\newblock The geometry of algorithms with orthogonality constraints.
\newblock {\em SIAM Journal on Matrix Analysis and Applications},
  20(2):303--353, 1998.

\bibitem[ERT11]{cit:opt_chaos}
M.~Ercsey-Ravasz and Z.~Toroczkai.
\newblock Optimization hardness as transient chaos in an analog approach to
  constraint satisfaction.
\newblock {\em Nature Physics}, 7(12):966, 2011.

\bibitem[ESC13]{Cit:target_tracking}
B.~Englot, T.~Sahai, and I.~Cohen.
\newblock Efficient tracking and pursuit of moving targets by heuristic
  solution of the traveling salesman problem.
\newblock In {\em 52nd IEEE Conference on Decision and Control}, pages
  3433--3438. IEEE, 2013.

\bibitem[FBR87]{Cit:qap_trace}
G.~Finke, R.~E. Burkard, and F.~Rendl.
\newblock Quadratic assignment problems.
\newblock {\em North-Holland Mathematics Studies}, 132:61--82, 1987.

\bibitem[Fie73]{Cit:Fiedler1}
M.~Fiedler.
\newblock Algebraic connectivity of graphs.
\newblock {\em Czechoslovak mathematical journal}, 23(2):298--305, 1973.

\bibitem[Fie89]{Cit:Fiedler2}
M.~Fiedler.
\newblock Laplacian of graphs and algebraic connectivity.
\newblock {\em Banach Center Publications}, 25(1):57--70, 1989.

\bibitem[GD04]{Cit:Procrustes-book}
J.C. Gower and G.B. Dijksterhuis.
\newblock {\em Procrustes problems}, volume~3.
\newblock Oxford University Press New York, 2004.

\bibitem[GJR91]{Cit:circuit-boards}
M.~Gr{\"o}tschel, M.~J{\"u}nger, and G.~Reinelt.
\newblock Optimal control of plotting and drilling machines: a case study.
\newblock {\em Mathematical Methods of Operations Research}, 35(1):61--84,
  1991.

\bibitem[GLO03]{Cit:sonet_rings}
O.~Goldschmidt, A.~Laugier, and E.V. Olinick.
\newblock {SONET/SDH} ring assignment with capacity constraints.
\newblock {\em Discrete Applied Mathematics}, 129(1):99--128, 2003.

\bibitem[Hel98]{Hel98}
K.~Helsgaun.
\newblock An effective implementation of the {L}in--{K}ernighan traveling
  salesman heuristic.
\newblock \textsc{Datalogiske Skrifter} ({W}ritings on {C}omputer {S}cience),
  {N}o. 81, Roskilde University, 1998.

\bibitem[Hel06]{Cit:keld2}
Keld Helsgaun.
\newblock {\em An effective implementation of K-opt moves for the Lin-Kernighan
  TSP heuristic}.
\newblock PhD thesis, Roskilde University. Department of Computer Science,
  2006.

\bibitem[Hel09]{Cit:helsgaun_k_opt}
Keld Helsgaun.
\newblock General k-opt submoves for the {L}in--{K}ernighan {TSP} heuristic.
\newblock {\em Mathematical Programming Computation}, 1(2-3):119--163, 2009.

\bibitem[Hel14a]{Cit:tsp_clustered}
K.~Helsgaun.
\newblock Solving the bottleneck traveling salesman problem using the
  {L}in--{K}ernighan--{H}elsgaun algorithm.
\newblock Technical report, Technical Report, Computer Science, Roskilde
  University, 2014.

\bibitem[Hel14b]{Cit:generalized_tsp}
K.~Helsgaun.
\newblock Solving the equality generalized traveling salesman problem using the
  {L}in--{K}ernighan--{H}elsgaun algorithm.
\newblock {\em Computer Science Report}, 141, 2014.

\bibitem[Hig86]{High86}
N.~J. Higham.
\newblock Computing the polar decomposition — with applications.
\newblock {\em {SIAM Journal on Scientific and Statistical Computing}},
  7(4):1160--1174, 1986.

\bibitem[HK70]{HK70}
M.~Held and R.~M. Karp.
\newblock The traveling-salesman problem and minimum spanning trees.
\newblock {\em Operations Research}, 18(6):1138--1162, 1970.

\bibitem[HK71]{HK71}
M.~Held and R.~M. Karp.
\newblock The traveling-salesman problem and minimum spanning trees: Part {II}.
\newblock {\em Mathematical Programming}, 1:6--25, 1971.

\bibitem[HK99]{HK99}
B.~R. Hunt and V.~Y. Kaloshin.
\newblock Regularity of embeddings of infinite-dimensional fractal sets into
  finite-dimensional spaces.
\newblock {\em Nonlinearity}, 12(5):1263--1275, 1999.

\bibitem[HRW90]{cit:qap_bounds}
S.~W. Hadley, F.~Rendl, and H.~Wolkowicz.
\newblock Bounds for the quadratic assignment problems using continuous
  optimization techniques.
\newblock In {\em IPCO}, pages 237--248, 1990.

\bibitem[HRW92]{HRW92}
S.W. Hadley, F.~Rendl, and H.~Wolkowicz.
\newblock A new lower bound via projection for the quadratic assignment
  problem.
\newblock {\em Mathematics of Operations Research}, 17:727--739, 1992.

\bibitem[Kar10]{TSP-NP-hard}
R.M. Karp.
\newblock Reducibility among combinatorial problems.
\newblock {\em 50 Years of Integer Programming 1958--2008}, pages 219--241,
  2010.

\bibitem[KB57]{KB57}
T.~C. Koopmans and M.~Beckmann.
\newblock Assignment problems and the location of economic activities.
\newblock {\em Econometrica}, 25(1):53--76, 1957.

\bibitem[KK07]{Cit:telescope2}
E.~Kolemen and N.J. Kasdin.
\newblock Optimal trajectory control of an occulter-based planet-finding
  telescope.
\newblock {\em American Astronautical Society}, pages 07--037, 2007.

\bibitem[KS18]{Cit:klus2014spectral}
S.~Klus and T.~Sahai.
\newblock A spectral assignment approach for the graph isomorphism problem.
\newblock {\em Information and Inference: A Journal of the IMA}, 7(4):689--706,
  2018.

\bibitem[LK73]{Cit:LinKernighan}
S.~Lin and B.W. Kernighan.
\newblock An effective heuristic algorithm for the traveling-salesman problem.
\newblock {\em Operations research}, 21(2):498--516, 1973.

\bibitem[LK81]{Cit:vehicle_routing}
J.K. Lenstra and A.H.G. Kan.
\newblock Complexity of vehicle routing and scheduling problems.
\newblock {\em Networks}, 11(2):221--227, 1981.

\bibitem[PB88]{PB88}
J.~C. Platt and A.~H. Barr.
\newblock Constrained differential optimization for neural networks.
\newblock Technical report, California Institute of Technology, 1988.

\bibitem[PP08]{PP08}
K.~B. Petersen and M.~S. Pedersen.
\newblock The matrix cookbook, 2008.

\bibitem[PR91]{Cit:cuttingplane}
M.~Padberg and G.~Rinaldi.
\newblock A branch-and-cut algorithm for the resolution of large-scale
  symmetric traveling salesman problems.
\newblock {\em SIAM review}, 33(1):60--100, 1991.

\bibitem[Rei91]{Cit:tsplib}
G.~Reinelt.
\newblock {TSPLIB}--{A} traveling salesman problem library.
\newblock {\em ORSA journal on computing}, 3(4):376--384, 1991.

\bibitem[SBC14]{cit:diff}
W.~Su, S.~Boyd, and E.~Candes.
\newblock A differential equation for modeling {N}esterov's accelerated
  gradient method: Theory and insights.
\newblock In {\em Advances in Neural Information Processing Systems}, pages
  2510--2518, 2014.

\bibitem[Sch68]{Sch68}
P.~Sch{\"o}nemann.
\newblock On two-sided orthogonal {P}rocrustes problems.
\newblock {\em Psychometrika}, 33(1):19--33, 1968.

\bibitem[SK15]{Cit:Tuhin_BN}
T.~Sahai and S.~Klus.
\newblock Automatic learning of {B}ayesian networks, May 2015.
\newblock US Patent App. 14/546,392.

\bibitem[SSB10]{sahai2010wave}
T.~Sahai, A.~Speranzon, and A.~Banaszuk.
\newblock Wave equation based algorithm for distributed eigenvector
  computation.
\newblock In {\em Decision and Control (CDC), 2010 49th IEEE Conference on},
  pages 7308--7315. IEEE, 2010.

\bibitem[SSB12]{Cit:sahai_hearing}
T.~Sahai, A.~Speranzon, and A.~Banaszuk.
\newblock Hearing the clusters of a graph: A distributed algorithm.
\newblock {\em Automatica}, 48(1):15--24, 2012.

\bibitem[SSO17]{cit:schafer2017owhadi}
F.~Sch{\"a}fer, T.J. Sullivan, and H.~Owhadi.
\newblock Compression, inversion, and approximate pca of dense kernel matrices
  at near-linear computational complexity.
\newblock {\em arXiv preprint arXiv:1706.02205}, 2017.

\bibitem[Ste87]{Ste87}
W.~R. Stewart.
\newblock Accelerated branch exchange heuristics for symmetric traveling
  salesman problems.
\newblock {\em Networks}, 17(4):423--437, 1987.

\bibitem[Won95]{Won95}
W.~S. Wong.
\newblock Matrix representation and gradient flows for {NP}-hard problems.
\newblock {\em Journal of Optimization Theory and Applications}, 87:197--220,
  1995.

\bibitem[WWJ16]{cit:wibisono}
A.~Wibisono, A.~C. Wilson, and M.~I Jordan.
\newblock A variational perspective on accelerated methods in optimization.
\newblock {\em Proceedings of the National Academy of Sciences},
  113(47):E7351--E7358, 2016.

\bibitem[WY10]{WY10}
Z.~Wen and W.~Yin.
\newblock A feasible method for optimization with orthogonality constraints.
\newblock Technical report, Rice University, 2010.

\bibitem[ZDG18]{ZDG18}
A.~Ziessler, M.~Dellnitz, and R.~Gerlach.
\newblock The numerical computation of unstable manifolds for infinite
  dimensional dynamical systems by embedding techniques.
\newblock {\em {Submitted to SIAM Journal on Applied Dynamical Systems
  (arXiv:1808.08787)}}, 2018.

\bibitem[ZP08]{ZP08}
M.~M. Zavlanos and G.~Pappas.
\newblock A dynamical systems approach to weighted graph matching.
\newblock {\em Automatica}, 44:2817--2824, 2008.

\end{thebibliography}
\end{document}